\pdfoutput=1
\documentclass[nonacm,acmsmall,screen,balance=false]{acmart}
\usepackage{jon-tikz}
\usepackage{iblock}
\usepackage{jmsdelim}
\usepackage{jmsoperators}
\usepackage{preamble}
\usepackage{macros}
\usepackage{categories}

\usetikzlibrary{decorations.markings,decorations.pathreplacing,
  shapes.geometric,matrix,calc,arrows,chains,positioning,scopes}

\pgfdeclarearrow{
  name = pxto,
  setup code = {
    \pgfarrowssettipend{1.5\pgflinewidth}
    \pgfarrowssetbackend{-2.5508\pgflinewidth}
    \pgfarrowssetlineend{-.25\pgflinewidth}
    \pgfarrowssetvisualbackend{-0.021\pgflinewidth}
    \pgfarrowsupperhullpoint{1.5\pgflinewidth}{0\pgflinewidth}
    \pgfarrowsupperhullpoint{-2.0085\pgflinewidth}{3.6525\pgflinewidth}
    \pgfarrowsupperhullpoint{-2.5508\pgflinewidth}{3.0763\pgflinewidth}
  },
  drawing code = {
    \pgfsetdash{}{0pt}%
    \pgfpathmoveto{\pgfpoint{1.5\pgflinewidth}{0.0254\pgflinewidth}}%
    \pgfpathlineto{\pgfpoint{-2.0085\pgflinewidth}{3.6525\pgflinewidth}}%
    \pgfpathlineto{\pgfpoint{-2.5508\pgflinewidth}{3.0763\pgflinewidth}}%
    \pgfpathlineto{\pgfpoint{-0.4322\pgflinewidth}{0.5\pgflinewidth}}%
    \pgfpathlineto{\pgfpoint{-0.4322\pgflinewidth}{-0.5\pgflinewidth}}%
    \pgfpathlineto{\pgfpoint{-2.5508\pgflinewidth}{-3.0763\pgflinewidth}}%
    \pgfpathlineto{\pgfpoint{-2.0085\pgflinewidth}{-3.6525\pgflinewidth}}%
    \pgfpathclose%
    \pgfusepathqfill
  }
}
\tikzset{>=pxto}
\tikzcdset{arrow style=tikz}
\tikzset{arrow/.style={->}}

\title{The \texorpdfstring{$\infty$}{infinity}-category of \texorpdfstring{$\infty$}{infinity}-categories in simplicial type theory}

\author{Daniel Gratzer}
\orcid{0000-0003-1944-0789}
\email{gratzer@cs.au.dk}
\affiliation{
  \institution{Aarhus University}
  \country{Denmark}
}

\author{Jonathan Weinberger}
\orcid{0000-0003-4701-3207}
\email{jweinberger@chapman.edu}
\affiliation{
  \institution{Chapman University}
  \country{USA}
}

\author{Ulrik Buchholtz}
\orcid{0000-0002-5944-6838}
\email{ulrik.buchholtz@nottingham.ac.uk}
\affiliation{
  \institution{University of Nottingham}
  \country{United Kingdom}
}

\begin{document}

\begin{abstract}
  Simplicial type theory (\STT{}) was introduced by Riehl and Shulman to leverage homotopy type
  theory to prove results about $(\infty,1)$-categories. Initial work on simplicial type theory
  focused on ``formal'' arguments in higher category theory and, in particular, no
  non-trivial examples of $\infty$-category theory were constructible within \STT{}. More recent
  work has changed this state of affairs by applying techniques developed initially for cubical type
  theory to construct the $\infty$-category of spaces. We complete this process by constructing the
  $\infty$-category of $\infty$-categories, recovering one of the main foundational results
  of $\infty$-category theory
  (straightening--unstraightening) purely type-theoretically. We also show how this construction
  enables new examples of the directed version of the structure identity principle: the
  structure homomorphism principle.
\end{abstract}
\maketitle

\begin{acks}
Jonathan Weinberger is grateful to the Fowler of School of Engineering at Chapman University for generous support of this work. He is particularly thankful to the Fletcher Jones Foundation and their award of a Fletcher Jones Foundation Faculty Fellowship in Engineering '25--'28 and the ensuing generous funding of this work. He also thanks the Schmid School of Science and Technology as well as the Center of Excellence in Computation, Algebra, and Topology (CECAT), both at Chapman University, for providing an excellent research environment.
\end{acks}

\section{Introduction}
\label{sec:intro}

A defining characteristic of dependent type theories is their focus on universes of (small) types.
More than in other foundations of mathematics, such universes are critical for even proving such basic
properties as $0 \neq 1$~\citep{smith:1989}. This focus on universes is only intensified with
\emph{homotopy type theory} (\HOTT{})~\citep{HoTTbook} where the universe is supplemented with the
univalence axiom. Such \emph{univalent} universes allow type theorists to view types as a synthetic
incarnation of spaces (\ie{}, $\infty$-groupoids) with the intensional identity type $a =_A b$
modeling paths in the type $A$. The utility of this viewpoint is demonstrated by development of
synthetic homotopy theory inside of \HOTT{}: a reconstruction of classical results in homotopy theory
with simpler and more conceptual proofs.

A long-standing challenge in \HOTT{} has been to broaden the reach of synthetic homotopy theory to
include homotopy-coherent algebraic structures and, especially, the homotopical enhancement of 
category theory: $\prn{\infty,1}$-category theory.\footnote{In both the title and hereafter, we
  shall simply write ``$\infty$-category'' or even ``category'' to refer to $\prn{\infty,1}$-categories. If
we wish to specifically discuss ordinary categories, we shall specifically denote them by 1-categories.} While
numerous approaches to this problem have been proposed~%
\citep{%
  licata:2011,%
  warren:2013,%
  nuyts:2015,%
  riehl:2017,%
  north:2018,%
  kavvos:directed:2019,%
  nuyts:2020,%
  weaver:2020,%
  ahrens:2023,%
  weaver:phd,%
  neumann:2024,%
  neumann:2025%
}
we will focus on the approach
introduced by \citet{riehl:2017}. There the authors leveraged a non-standard model
of \HOTT{} where types are realized by simplicial spaces. In particular, they showed
that the complete Segal spaces---a known model of $\infty$-categories~\citep{rezk:2001}---then
arise as certain types satisfying a pair of easily-defined properties. Thus, in this setting not
every type is an $\infty$-category, but every $\infty$-category gives rise to a valid type.

Concretely, \emph{simplicial type theory} extends \HOTT{} with a \emph{directed interval}, a postulated
totally ordered lattice $\prn{\Int,0,1,\le}$. This new type is meant to represent the $\infty$-category
with two objects $0,1$ and a single non-identity morphism $0 \to 1$---an interpretation justified by
the model of \STT{} in simplicial spaces---and we then use $\Int$ to define morphisms in an arbitrary
type $A$ as ordinary functions $\Int \to A$. By constraining the endpoints of a synthetic morphism,
we arrive at the definition of the space of synthetic morphisms in a type:
$\Hom[A]{a}{b} = \Sum{f : \Int \to A} f\,0 = a \times f\,1 = b$.

\citet{riehl:2017} then demonstrate that the definition
of an $\infty$-category can be formulated concisely as a \emph{predicate} $\IsCat$ on types,
essentially requiring every pair of composable morphisms have a unique composite. Furthermore,
they show that ordinary functions between such types constitute functors and that other classical
definitions in $\infty$-category theory become expressible. Subsequent work has further expanded
this approach, developing fibered category theory~\citep{buchholtz:2023,weinberger:phd}, limits
and colimits~\citep{bardomiano:2025:limits}, \etc{}

While not every type constitutes an $\infty$-category in \STT{}, many type-theoretic operations
preserve the property of being an $\infty$-category. For instance, $\ObjInit{}$ and $\ObjTerm{}$ are
the initial and terminal categories, $A \times B$ ($A + B$) is the (co)product category, and $A \to
B$ is the category of functors. As an extension of \HOTT{}, \STT{} comes equipped with a (hierarchy of)
universes and it is therefore natural to ask:
\begin{center}
  \em
  Is $\Uni$ a recognizable category, \eg{}, the category of categories?
\end{center}
Unfortunately, the answer is negative; $\Uni$ is the canonical example of a type that is \emph{not} an
$\infty$-category in \STT{}. In fact, even if one considers simple subtypes of the universe
(\eg{}, $\Sum{A : \Uni} \IsCat\prn{A}$) one does not obtain a category, as synthetic morphisms $\Int
\to \Sum{A : \Uni} \IsCat\prn{A}$ neither compose nor faithfully represent functors. However, it has
long been conjectured that the category of categories should be constructible in \STT{} as a certain
subtype of the universe.

We address this final gap in the foundations of \STT{} by settling this conjecture
affirmatively and constructing the category of categories as a subtype ${\Cat} \hookrightarrow
{\Uni}$ and verifying its essential properties.

\subsection{Directed univalence and \texorpdfstring{$\Cat$}{the category of categories}}

What criteria should be used to determine if a subtype ${\Cat} \hookrightarrow {\Uni}$ is a valid definition
of the category of categories? If one is not considering $\infty$-categories, the answer to this
question is straightforward: $\Cat$ is a valid definition if the objects denote precisely small
types satisfying $\IsCat$, synthetic morphisms are exactly functions (\ie{}, functors) between these
types, and the composition and identity operations behave as expected. In the $\infty$-categorical
case, the story does not end here; we must also convince ourselves that all the higher synthetic
morphisms also behave ``as expected''. However, it is far from clear what the expected behavior
ought to be! Instead, we take a different approach by isolating a particular universal property 
for the embedding $\Cat \hookrightarrow \Uni$ and arguing that all the other properties of $\Cat$
stem from this single property.

The universal property in question is an $\infty$-categorical version of the Grothendieck
correspondence, often referred to as straightening--unstraightening in the $\infty$-categorical
context. Classically, the Grothendieck correspondence states there is an equivalence between
pseudofunctors $\mathcal{C} \to \CAT$ for some 1-category $\mathcal{C}$ and pairs of a category
$\mathcal{E}$ and a cocartesian family $\mathcal{E} \to \mathcal{C}$.\footnote{In the 1-
  and 2-categorical literature these are often referred to as Grothendieck opfibrations. We
use ``cocartesian'' as it is the standard term in $\infty$-category theory.} Naively, this suggests a
simple universal property for the type $\Cat$: We should require for every $\infty$-category $C$
an equivalence between the type of functions $C \to \Cat$ and the type
$\Sum{E : \Uni}\Sum{\pi : E \to C} \IsCocart\prn{\pi}$. Unfortunately, this definition is too
straightforward: if such a $\Cat$ existed we would be able to show that $\IsCocart$ held for every
map $\pi$ which is certainly not true.

The problem is familiar to cubical type theorists: we are essentially asking for $\Cat$ to
be a \emph{universal fibration}, where $\IsCocart$ is our notion of fibrancy. It is well-known that
one cannot give  universal property for the univalent universe in the extensional type theory of cubical sets~\citep{orton:2018}: the naive equivalence is only valid with respect to \emph{closed}
elements and cannot be internalized this way into type theory where it would also apply to elements
in an arbitrary context. A solution to this problem in cubical type theory was presented
by \citet{licata:2018}. There the authors give a description of a universal fibration in an extensional type theory for cubical sets using \emph{modal} type theory. They supplement type theory with an idempotent
comonad $\GM$, such that the modal type $\Modify[\GM]{A}$ contains only the \emph{global} elements
of $A$. After further extending type theory with a right adjoint to $\Int \to -$, they are then not only able to express the desired universal property (appropriately annotated with $\GM$), but
also use the aforementioned right adjoint to derive it.

This idea has been adapted to variants of simplicial type theory first by \citet{weaver:2020} and
later by \citet{gratzer:2024} to construct a category of \emph{discrete} $\infty$-categories
(spaces). We follow on this line of work---specifically \citet{gratzer:2024}---and use the $\GM$
modality to state the correct universal property for the category of categories and produce a type
satisfying it. In particular:
\begin{restatable*}{theorem}{univcocart}\label{thm:cat:univ-cocart}
  $\Cat$ is the base of the universal cocartesian family, \ie{}, for any $\DeclVar{C}{\GM}{\Uni}$, we
  have $\Modify[\GM]{C \to \Cat} \Equiv \Modify[\GM]{\Sum{E : C \to \Uni} \IsCocart\prn{E}}$.
\end{restatable*}
On its own, this statement is not worth much. We do not know whether $\Cat$ is a category itself,
let alone to what its objects and morphisms correspond. To resolve this, we give a detailed analysis
of the structure of cocartesian fibrations over certain categories, relying on the established
results in simplicial type theory. As a particular case of these results, we conclude that
cocartesian fibrations over $\Int$ are determined by (1) a pair of categories $C$, $D$ along with (2)
a function $C \to D$. Consequently, we are able to prove the directed univalence theorem for $\Cat$: 
\begin{restatable*}[Directed univalence]{corollary}{dua}\label{cor:dua:dua}
  If $A,B$ are $\GM$-elements of $\Cat$, then there is a equivalence
  $\Con{dua} : \Hom[\Cat]{A}{B} \Equiv \Modify[\GM]{A \to B}$.
\end{restatable*}
Here we see the first major divergence between the case of categories and the previously
considered case of discrete categories. Directed univalence for $\Cat$ only guarantees that
$\Hom[\Cat]{A}{B}$ is equivalent to $\Modify[\GM]{A \to B}$ whereas in the discrete case one
obtains an equivalence with $A \to B$. This is inevitable: \emph{no} category
can have hom-spaces equivalent to $A \to B$ for arbitrary categories $A,B$. However, the
appearance of modalities in directed univalence is a significant complication.

Further consequences of our analysis prove that $\Cat$ is a category and that composition of
synthetic morphisms corresponds (under directed univalence) to composition of ordinary
functions. In total then, we show that the universal property of $\Cat$ is sufficient to derive all
the expected behavior of a category of categories as well as uniquely characterize $\Cat$ itself.

Directed univalence opens up a wide array of applications. For instance, we use it to show that
$\Cat$ contains a number of recognizable \emph{reflective} subcategories (truncated categories,
partial orders, spaces, \etc{}). It also allows us to build new and important categories by
combining $\Cat$ with existing type-theoretic connectives. As a simple example, we show that
$\Sum{c : \Cat} c$ can be analyzed using directed univalence and use its attendant
\emph{structure homomorphism principle} (SHP) to see that $\Sum{c : \Cat} c$ is the \emph{lax}
slice category $\ObjTerm{}\dblslash\Cat$.

Finally, while \cref{thm:cat:univ-cocart} is a form of the Grothendieck correspondence, it is not
the strongest form possible. This would state that there is an equivalence of
\emph{$\infty$-categories} between cocartesian families over $C$ and the functor category
$C \to \Cat$. We adapt an argument sketched by \citet{cisinski:2024} to \STT{} to show that
\cref{thm:cat:univ-cocart} upgrades to an equivalence of categories and thereby give
a synthetic proof of the \emph{straightening-unstraightening} theorem~\citep{lurie:2009}.

\subsection{Contributions and outline}

Our main contribution is the first construction of a directed univalent category of categories
within \STT{} together with a verification of all its essential properties. This is the last major
primitive missing from the theory of $\infty$-categories in simplicial type theory and,
additionally, also constitutes a new and novel approach to a foundational theorem in
$\infty$-category theory itself. With this in place, we also note that a putative full directed type
theory can be interpreted into $\infty$-categories by giving a mere syntactic translation to the
appropriate fragment of \STT{}.

\begin{itemize}
  \item We extend the methodology of \citet{licata:2018} to construct a universal cocartesian family
    $\Cat_\bullet \to \Cat$.
  \item We derive purely from this universal property that $\Cat$ is a category, satisfies directed
    univalence, \etc{}
  \item We prove various properties of $\Cat$ and showcase the novel applications of SHP it enables.
  \item Finally, we give a fully synthetic proof of the foundational straightening--unstraightening
    theorem.
  \end{itemize}

Over all, we demonstrate how type theory (through \STT{}) is highly effective for proving major
results in higher category theory.

The remainder of the paper is organized as follows. In \cref{sec:primer} we recall the main
ideas of simplicial type theory, its triangulated type theory variant, and the modal extensions
thereof. This, along with \cref{sec:cocartesian} on cocartesian fibrations, is intended to make the paper as
self-contained as possible. In \cref{sec:cat} we define $\Cat$ and prove \cref{thm:cat:univ-cocart},
the main construction of this paper. In \cref{sec:dua}, we combine our new results on cocartesian
fibrations along with the results of the prior section to prove \cref{cor:dua:dua} along with the
fact that $\Cat$ is a category. Finally, in \cref{sec:st-unst,sec:examples}, we give a
number of consequences of our results, including a proof of straightening--unstraightening
(\cref{cor:st-unst:st-unst}).

\section{Simplicial and triangulated type theory}
\label{sec:primer}

We now turn to giving a more precise account of the flavor of simplicial type theory we shall use. There are two major modifications to the theory compared with the original system
studied by \citet{riehl:2017}. First is that we supplement our theory with a collection of
modalities, including the $\GM$ modality discussed in the introduction. We do this using
\MTT{}, a general purpose modal dependent type theory~\citep{gratzer:2020}. Secondly, the intended
model of \STT{} is simplicial spaces, the $\infty$-categorical version of $\PSH{\Delta}$. However,
in order to leverage the techniques of \citet{licata:2018} to define $\Cat$, we require a
postulated \emph{amazing right adjoint} to $\Int \to -$. Unfortunately, in the intended model of
simplicial spaces, such a right adjoint simply does not exist; the category of simplices does not
have products.

To circumvent this, \citet{gratzer:2024} introduced a further relaxation of the modal variant of \STT{}:
\emph{triangulated type theory} \TTT{}. The intended model of \TTT{} is \emph{Dedekind cubical} spaces,
rather than simplicial ones. Simplicial spaces embed into this intended model as a subtopos and the
original structure therefore remains. Concretely, this change involves relaxing the requirement that
$\Int$ be totally ordered to merely asking it to be a bounded distributive lattice. In this
section, we recall the details of \TTT{} relevant for this work and, in particular, describe the
modal extensions necessary to construct the category of categories.

\subsection{First steps in triangulated type theory}

As mentioned in \cref{sec:intro}, \STT{} and its relaxation \TTT{}, extend homotopy type theory.
Accordingly, we begin by recalling the basic definitions and notations from homotopy type theory.
For a more complete account, see the \citet{HoTTbook}.

Homotopy type theory begins with intensional Martin-L{\"o}f type theory, complete with a hierarchy of universes, $\Uni[0], \Uni[1], \dots$, and
intensional identity types $a =_A b$. We use $\Prod{}$ and $\Sum{}$ for the standard dependent product and sum types of this type theory and
use $p_!$ for the transport map $B\prn{x} \to B\prn{y}$ induced by $p : x = y$. The key extension of
\HOTT{}, the univalence axiom, governs the behavior of the intensional identity type of the
universe. In particular, this axiom relates equality in the universe to \emph{equivalences}. See \opcit{} for a detailed discussion of the type of equivalences $A \Equiv B$; we only
note that it is a subtype of $A \to B$ which, when $A = B$, includes the identity function.
\begin{restatable}{axiom}{univalence}
  The canonical map $A =_{\Uni[i]} B \to A \Equiv B$ sending $\Refl$ to $\ArrId{}$ is an
  equivalence.
\end{restatable}

We shall also have use for a few other \HOTT{} concepts. First, since univalence causes the
\emph{unicity of identity proofs} to fail rather spectacularly, so we isolate those types for which
this property still holds as (homotopy) \emph{sets}. We will also use for the stronger properties of
being a (homotopy) \emph{proposition}, or \emph{contractible}:
\begin{align*}
  &\IsSet\prn{A} = \Prod{a,b : A} \IsProp\prn{a = b}
  \\
  &\IsProp\prn{A} = \Prod{a,b : A} \IsContr\prn{a = b}
  \\
  &\IsContr\prn{A} = \Sum{a : A} \Prod{b : A} a = b
\end{align*}
If $\phi : \Uni \to \Uni$ is valued in propositions, we write $\Uni[\phi]$ for
$\Sum{A : \Uni} \phi\prn{A}$. Moreover, we use $\Prop$ as condensed notation for
$\Uni[\IsProp]$.

Next, we require a few \emph{higher inductive types} (HITs). These are inductive types
that postulate constructors not only of elements of the type, but also of elements of its identity
type. For instance, we shall use pushouts $A +_B C$ and localizations as defined by
\citet{rijke:2020}. Since the models of HITs satisfying strict equations
is somewhat fraught, we only assume that our HITs satisfy propositional computation rules, \ie{}, that
they are homotopy-initial algebras.

We now turn to the first and most basic axiom of triangulated type theory, which
introduces the directed interval.
\begin{restatable}{axiom}{intax}
  There is a set $\Int$ equipped with the structure of a bounded distributive lattice
  $\prn{0,1,\land,\lor}$ such that $0 \neq 1$.
\end{restatable}

Crucially, compared with \STT{} we do not assume that $\Prod{i, j : \Int} i \le j \lor j \le i$
holds. However, just as with the failure of UIP, we may isolate those types that ``believe $i \le j \lor
j \le i$'' as \emph{simplicial} types. These are types $A$ such that the following
holds:
\[
  \IsSimp\prn{A} = \Prod{i,j : \Int} \IsEquiv\prn{\Const : A \to \prn{i \le j \lor j \le i \to A}}
\]
The next result follows from \citet{rijke:2020}:
\begin{proposition}
  There is a lex idempotent monad $\Simp : \Uni \to \Uni$ with $\Simp A$ universal
  among simplicial types receiving a map from $A$.
\end{proposition}
In other words, ordinary simplicial type theory is the ``subset of \TTT{}'' given by types for which
$\eta_A : A \to \Simp A$ is an equivalence. We write $\Uni[\Simp] = \Sum{A : \Uni} \IsEquiv\prn{\eta_A}$ for
the subtype of $\Uni$ spanned by simplicial types and note that $\Uni[\Simp]$ is itself simplicial.
Consequently, there is a unique $\DepSimp : \Simp \Uni \to \Uni[\Simp]$ such that
$\DepSimp \circ \eta = \Simp$.

With $\Int$, we can make a number of important definitions that isolate those types that actually
behave like categories. First, as mentioned earlier, we can define a \emph{synthetic morphism}
in a type $A$ to be a map $\Int \to A$. Constraining the endpoints of this map, we arrive at the
definition of the synthetic hom-type:
\[
  \Hom[A]{a}{b} = \Sum{f : \Int \to A} f\,0 = a \times f\,1 = b
\]

Every type comes equipped with a collection of identity morphisms: $\ArrId{a} = \lambda \_.\,a :
\Hom{a}{a}$. However, the other operation we expect from a category---composition---is not available
by default. What is available instead is a proof-relevant composition \emph{relation}, stating that
a pair of synthetic morphisms compose to a third. To define these, we derive the $2$-simplex
$\Delta^2$ from $\Int$. For the sake of future use, we define the $n$-simplex for any $n\ge0$:
\[
  \Delta^n = \Sum{\vec{i} : \Int^n} \prn{i_0 \ge i_1 \ge \dots \ge i_{n - 1}}
\]
Crucially, the 2-simplex can be visualized as a triangle (the shaded region below), if $\Int$ is
seen as a line segment whence $\Int^2$ as a square:
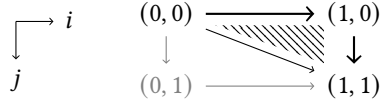
\[
  \begin{tikzpicture}[diagram]
    \draw[->] (-2,-.1) -- (-1.5,-.1) node[right] {$i$};
    \draw[->] (-2,-.1) -- (-2,-.6) node[below] {$j$};
    \node (A) at (0,0) {$(0,0)$};
    \node[right = 2.5cm of A] (B) {$(1,0)$};
    \node[color = gray, below = 0.95cm of A] (C) {$(0,1)$};
    \node[below right = 0.95cm and 2.5cm of A] (D) {$(1,1)$};
    \fill[pattern=north west lines] (0.65,-0.15) -- (2.1,-0.15)  -- (2.1,-0.65) -- cycle;
    \path[->, thick] (A) edge (B);
    \path[->, color = gray] (A) edge (C);
    \path[->, color = gray] (C) edge (D);
    \path[->, thick] (B) edge (D);
    \path[->] (A) edge (D);
  \end{tikzpicture}
\]
We also have need for the inner horn $\Lambda^2_1 = \Sum{i,j : \Int} i = 1 \lor j = 0$, a
subtype of $\Delta^2$ which captures the two highlighted segments above.

Say that $\tau : \Delta^2 \to A$ is a witness for the fact that $\tau\prn{-,0}, \tau\prn{1,-} :
\Int \to A$ compose to $\lambda i.\,\tau\prn{i,i}$. A Segal type (sometimes called a pre-category)
is one where every pair of composable arrows admits a unique composition witness. To define this
formally, we note that maps $\Lambda^2_1 \to A$ precisely capture the data of composable
arrows:
\begin{definition}
  A type $A$ is \emph{Segal} if $\IsEquiv\prn{A^{\Delta^2} \to A^{\Lambda^2_1}}$ holds.
\end{definition}
\begin{notation}
  If $f,g : A^\Int$ and $p : f\prn{1} = g\prn{0}$, we write $\brk{f,g,p}$ for the induced map
  $A^{\Lambda^2_1}$ and, if $A$ is Segal, $g \circ_p f$ for the composite. Furthermore, we shall
  subsequently have use for the outer horns $\Lambda^2_0 = \Sum{i,j : \Int} i = j \lor j = 0$ and
  $\Lambda^2_2 = \Sum{i,j : \Int} i = 1 \lor i = j$. Finally, we write $\bar{i}$ for the element of
  $\Delta^{n + i}$ given by $\prn{1, \dots 1, 0, \dots}$ of $i$ copies of $1$ followed by $n$ copies
  of $0$.
\end{notation}

Segal types enjoy a unique composition operation given by the inverse of the map
$A^{\Delta^2} \to A^{\Lambda^2_1}$, and calculation shows that the aforementioned definition of the
identity morphism is a left and right unit for composition. However, objects in a pre-category have
two distinct notions of sameness: via either the identity type \emph{or} synthetic isomorphism. By the
latter, we mean a morphism $f : \Hom{a}{b}$ equipped with $g,h : \Hom{b}{a}$ along with composition
witnesses showing that $g$ ($h$) is left (right) inverse to $f$. One can define
$\IntIso = \Delta^2 \Pushout{\Lambda^2_2} \Int \Pushout{\Lambda^2_0} \Delta^2$ such that
$\IntIso \to X$ precisely corresponds to an equivalence in $X$~\citep[\textsection 4.2]{buchholtz:2023}). A
distinctive feature of $\infty$-category theory is that these two notions (object equality and
isomorphism) can be made to coincide; a property similar to the univalence axiom. We therefore also
single out those types which satisfy this local univalence condition:
\begin{definition}
  A type $A$ is \emph{Rezk} if $\IsEquiv\prn{\Con{const} : A \to A^{\IntIso}}$.
\end{definition}

\begin{definition}
  A simplicial, Segal, and Rezk type is called a \emph{category}. A category whose
  morphisms are all invertible is a \emph{groupoid}.%
\end{definition}
\begin{remark}
  The general results of \citet{rijke:2020} show that categories and groupoids are modal types for
  idempotent monads. We write $\Grpdify$ for the idempotent modality associated with groupoids
  in particular, \ie{}, nullification at $\Int$.
\end{remark}

We shall also have occasion to use the \emph{relative} versions of the Segal and Rezk conditions.
Given a family of types $A : X \to \Uni$, we say that $A$ is (\emph{right}) \emph{orthogonal} to a map $I
\to J$ if the following canonical map is an equivalence:
\[
  \prn{\Sum{x : X} A\prn{x}}^J \to X^J \times_{X^I} \prn{\Sum{x : X} A\prn{x}}^I
\]
The relative Segal condition asks that a family of types $A : X \to \Uni$ be right orthogonal to
$\Lambda^2_1 \to \Delta^2$ and the relative Rezk condition asks the same for
$\IntIso \to \ObjTerm{}$. A Segal family is called \emph{inner} and a family that is both Segal and Rezk
is \emph{iso-inner}.

For use in \cref{sec:cat}, we note that we can phrase the requirement that a family be
inner using the following predicate:
\begin{align*}
  &\Con{isInner} : \prn{\Delta^2 \to \Uni} \to \Prop
  \\
  &\Con{isInner}\,A =
  \IsEquiv\prn{\prn{\Prod{t : \Delta^2} A\,t} \to \prn{\Prod{t : \Lambda^2_1} A\,t}}
\end{align*}
A family $\DeclVar{A}{}{X \to \Uni}$ is inner if and only if $\Prod{h : \Delta^2 \to X} \Con{isInner}\prn{A \circ h}$
holds.

\begin{notation}
  \label{not:primer:proj}
  We shall often identify a family $\DeclVar{A}{}{X \to \Uni}$ with its associated total
  space projection $\pi$ from $\TotalTy{A} \defeq \Sum{\DeclVar{x}{}{X}} A\,x$ to $X$.
  We shall say that an arbitrary map of types $f : X \to Y$ is, for instance, inner if the
  associated map $Y \to \Uni$ sending $y$ to $f^{-1}\prn{y}$ is inner.
\end{notation}

Given a family $\DeclVar{A}{}{X \to \Uni}$ and an arrow $\DeclVar{f}{}{\Int \to X}$ we define
\emph{dependent arrows} over $f$ from $\DeclVar{a}{}{A\prn{f\,0}}$ to
$\DeclVar{a'}{}{A\prn{f\,1}}$  as follows:
\[
  \hom_f^A(a,a') \defeq
  \Sum{\DeclVar{\alpha}{}{\prn{t : \Int} \to A(f\,t)}}
  \prn{\alpha\,0 = a} \times \prn{\alpha\,1 = a'}
\]
In an inner family there exists an induced composition operation for dependent arrows~\cite{buchholtz:2023}.

\subsection{Multimodal type theory}

The next step in \TTT{} is to include modalities: special type constructors that violate key
properties we ordinarily require in type theory, such as stability under substitution. We use
these modalities to internalize crucial operations from our intended model of cubical spaces such as
the discrete and codiscrete endofunctors, the opposite functor, \etc{} To this end, we recall
some of the details of the modal extension to type theory, \MTT{}, following~\citep{gratzer:2020}. See \citet{gratzer:phd} for a more detailed account. Since our primary
goal is to write programs in \MTT{}, we focus on the ``informal'' version of the syntax and
defer the formal rules (replete with de Bruijn indices and a substitution calculus) to
\cref{sec:mtt-syntax}.

First, \MTT{} is parameterized by a \emph{mode theory} $\Mode$. This is a strict 2-category describing
the modalities (as 1-cells) and transformations between them (as 2-cells). While \MTT{}
also permits distinct type theories to be related by modalities by considering mode theories with
multiple modes (0-cells), we do not need this generality and therefore assume that mode theories
have only one object. We shall also only be concerned with mode theories with at most one 2-cell
between every pair of 1-cells \ie{}, 2-categories that are merely poset-enriched. As such, our mode theories are simply given by ordered monoids. The mode theory
required for \TTT{} is described in \cref{sec:primer:modal-ttt}, but we continue with an
arbitrary mode theory satisfying these constraints for the moment.

The main extension of \MTT{} is to add a new modal type $\Modify[\mu]{-}$ for each modality $\mu \in
\Con{Arr}\prn{\Mode}$. However, as already mentioned modal types are somewhat peculiar, and to
accommodate them \MTT{} also modifies context extension. Specifically, each variable in an \MTT{}
context is annotated with a ``formal division'' of modalities $\DeclVar{x}{\mu/\nu}{A}$. We write
$\Gamma/\nu$ for the operation which modifies each annotation in $\Gamma$ to send
$\DeclVar{x}{\mu/\nu\/_0}{A}$ to $\DeclVar{x}{\mu/\nu\/_0 \circ \nu}{A}$. The variable rule is then
modified to account for these formal divisions as follows:
\begin{mathpar}
  \inferrule{
    \mu \le \nu
    \\
    \DeclVar{x}{\mu/\nu}{A} \in \Gamma
  }{
    \IsTm{x}{A}
  }
\end{mathpar}
Note that one can recover the ordinary rules for variables by considering the annotation
$\ArrId{}/\ArrId{}$. As a matter of notation therefore, we generally suppress division by $\ArrId{}$
and omit the annotation entirely for $\ArrId{}/\ArrId{}$ so that we instead write
$\DeclVar{x}{\mu}{A}$ or $\DeclVar{x}{}{A}$.

These annotations are then manipulated by the modal operators $\Modify[\mu]{-}$. In particular,
they are added by the formation and introduction rules. The (somewhat lengthy) elimination rule, on
the other hand, papers over the difference between annotations on a variable and modal types by
allowing us to convert a binding $\DeclVar{x}{\nu/\ArrId{}}{\Modify[\mu]{A}}$ into a binding
$\DeclVar{y}{\nu\circ\mu/\ArrId{}}{A}$:
\begin{mathpar}
  \inferrule{
    \IsTy[\Gamma/\mu]{A}
  }{
    \IsTy{\Modify{A}}
  }
  \and
  \inferrule{
    \IsTm[\Gamma/\mu]{a}{A}
  }{
    \IsTm{\MkMod{a}}{\Modify{A}}
  }
  \\
  \inferrule{
    \IsTy[\Gamma/\nu\circ\mu]{A}
    \\
    \IsTy[\Gamma, \DeclVar{y}{\nu/\ArrId{}}{\Modify[\mu]{A}}]{B\prn{y}}
    \\\\
    \IsTm[\Gamma, \DeclVar{x}{\nu\circ\mu/\ArrId{}}{A}]{b\prn{x}}{\Sb{B}{\MkMod[\mu]{x}/y}}
    \\
    \IsTm[\Gamma/\nu]{a}{\Modify[\mu]{A}}
  }{
    \IsTm{\LetMod<\nu>[\mu]{a}[x]{b\prn{x}}}{\Sb{B}{a/y}}
  }
  \and
  \LetMod<\nu>[\mu]{\MkMod{a_0}}[x]{b\prn{x}} = \Sb{b}{a_0/x}
\end{mathpar}

Aside from these modal types, \MTT{} extends type theory with an extended version of the dependent
product type, $\prn{\DeclVar{a}{\mu}{A}} \to B\prn{a}$. Up to equivalence, these modal dependent
products are the same as $\prn{\DeclVar{a}{}{\Modify[\mu]{A}}} \to
\LetMod[\mu]{a}[a_0]{B\prn{a_0}}$, but offer lighter-weight notation and support a stronger
$\eta$-rule.
\begin{mathparpagebreakable}
  \inferrule{
    \IsTy[\Gamma/\mu]{A}
    \\
    \IsTm[\Gamma, \DeclVar{a}{\mu}{A}]{b\prn{a}}{B\prn{a}}
  }{
    \IsTm{\lambda a.\,b\prn{a}}{\prn{\DeclVar{a}{\mu}{A}} \to B\prn{a}}
  }
  \and
  \inferrule{
    \IsTm{f}{\prn{\DeclVar{a}{\mu}{A}} \to B\prn{a}}
    \\
    \IsTm[\Gamma/\mu]{a}{A}
  }{
    \IsTm{f\prn{a}}{B\prn{a}}
  }
\end{mathparpagebreakable}
When we write ``for all $\DeclVar{a}{\mu}{A}$, the type $B\prn{a}$ is inhabited''
this should be understood to denote a function $\prn{\DeclVar{a}{\mu}{A}} \to B\prn{a}$.

Finally, we require the following technical axiom on \MTT{} which governs the interaction between
modalities and identity types:
\begin{restatable}{axiom}{crispind}
  \label{ax:crisp-ind}
  If $\DeclVar{A}{\mu}{\Uni}$ and $\DeclVar{a,b}{\mu}{A}$, then the following canonical map sending
  $\Refl$ to $\MkMod[\mu]{\Refl}$ is an equivalence:
  \[
     \MkMod[\mu]{a} = \MkMod[\mu]{b} \to \Modify[\mu]{a = b}
  \]
\end{restatable}

We summarize several basic properties of modalities:
\begin{proposition}\label{prop:primer:modalities}
  \leavevmode
  \begin{itemize}
    \item Each $\Modify[\mu]{-}$ commutes with $\ObjTerm{}$ and pullbacks.
    \item
      $\Modify[\mu]{\Modify[\nu]{-}} \Equiv \Modify[\mu\circ\nu]{-}$ and
      $\Modify[\ArrId{}]{-} \Equiv -$.
    \item If $\mu \le \nu$, then there is a canonical map
      $\Modify[\mu]{-} \to \Modify[\nu]{-}$.
    \item There is an equivalence
      $\Modify[\GM]{A \to \Modify[\SM]{B}} \Equiv \Modify[\GM]{\Modify[\GM]{A} \to B}$
      for $\DeclVar{A,B}{\GM}{\Uni}$. Similarly with $\SM$ and $\GM$ replaced with $\OM$.
  \end{itemize}
\end{proposition}
\noindent
From the first point, we obtain
$\ZApp : \Modify[\mu]{A \to B} \times \Modify[\mu]{A} \to \Modify{B}$.
We often write $f^\dagger$ for $\MkMod[\mu]{f} \ZApp -$, when the latter is
well-formed.

\subsection{Modalities in triangulated type theory}
\label{sec:primer:modal-ttt}

We now turn from \MTT{} generally to the specific instantiation of \MTT{} we require for \TTT{}.
Namely, we will work with the mode theory generated by one 0-cell $m$, three generating non-identity
1-cells $\brc{\GM,\SM,\OM}$, and the following (in)equalities:
\begin{mathpar}
  \GM = \GM \circ \GM = \GM \circ \SM = \GM \circ \OM
  \and
  \SM = \SM \circ \GM = \SM \circ \SM = \SM \circ \OM
  \and
  \GM \le \ArrId{} = \OM \circ \OM \le \SM
\end{mathpar}
For intuition, in the cubical spaces model of \TTT{} the modal operations
for the generating modalities are interpreted as follows:
\begin{itemize}
  \item $\Interp{\Modify[\GM]{-}}$ is the discrete functor---sending a cubical space $X$ to the
    constant cubical space $\Delta\prn{X\prn{\brk{0}}}$.
  \item $\Interp{\Modify[\SM]{-}}$ is the codiscrete functor---right adjoint to the discrete
    functor---and sends a cubical space $X$ to the cubical space whose value at $\brk{n}$
    is given by $X\prn{\brk{0}}^{2^n}$.
  \item $\Interp{\Modify[\OM]{-}}$ sends a cubical space $X$ to $X \circ \mathrm{op}$ where
    $\mathrm{op}$ is the functor on the cube category reversing $0$ and $1$.
\end{itemize}

If we consider a category $\DeclVar{C}{\GM}{\Uni}$, then these interpretations of $\GM$ and $\OM$
have very concrete meanings: $\Modify[\GM]{C}$ is the underlying discrete groupoid of objects
of $C$ and $\Modify[\OM]{C}$ is the opposite category. In general $\Modify[\SM]{C}$ will fail to be
a category even when $C$ is---its utility is mostly in forcing $\GM$ to be a left adjoint. We note
that the first two modalities are an instance of cohesion~\citep{lawvere:2007,shulman:2018,myers:2023}.

Note, however, that these are presently just intuitions from a particular model. To actually make
these have force within \TTT{} itself, we require additional axioms governing the behavior of
$\Int$ and its interactions with modalities. A complete list of axioms is given in \cref{sec:axioms}
so we discuss only axioms playing a direct role in the proofs given in the paper.

First, we have axioms governing the behavior of $\GM$ and $\Int$:
\begin{restatable}{axiom}{intdisc}
  \label{ax:int-global-points}
  $0,1 : \Int$ induce an equivalence $\Modify[\GM]{\Bool} \Equiv \Modify[\GM]{\Int}$
\end{restatable}
\begin{restatable}{axiom}{intdetectsdisc}\label{ax:int-detects-discrete}
  If $\DeclVar{A}{\GM}{\Uni}$, then $\IsEquiv\prn{\Modify[\GM]{A} \to A}$ if\/
  $\IsEquiv\prn{A \to A^\Int}$
\end{restatable}

Intuitively, the first of this pair states that $\Int$ has just two global elements: the specified
endpoints $0$ and $1$. The second axiom ensures that $\Int$ ``detects'' whether a type is of the
form $\Modify[\GM]{A}$. Phrased pithily: a type $A$ is equivalent to a groupoid of objects if and
only if it has no non-invertible synthetic morphisms. These axioms---particularly the last
one---are what force $\Modify[\GM]{-}$ to match our earlier idea of taking the groupoid core.

The next axiom forces cubes $\Int^n$ to form a \emph{separating family} among all types. That is, to
detect whether a particular ($\GM$-)map is invertible, it suffices to check whether post-composing
with this map induces equivalences between groupoids of maps $\Modify[\GM]{\Int \to -}$:
\begin{restatable}{axiom}{cubesseparate}
  \label{ax:cubes-separate}
  If $\DeclVar{A,B}{\GM}{\Uni}$ and $\DeclVar{f}{\GM}{A \to B}$, then $f$ is an equivalence if
  \[
    \Prod{n : \Nat}
    \IsEquiv\prn{\prn{f_*}^\dagger : \Modify[\GM]{\Int^n \to A} \to \Modify[\GM]{\Int^n \to B}}
  \]
\end{restatable}

Since $\Simp$ has a universal ``mapping-out'' property, we have a succinct description of
$\Simp A \to B$ in terms of $A \to B$. In general, there is no such simple relationship between
$A \to \Simp B$ and $A \to B$, with the notable exception of the case when $A = \Delta^n$.
Intuitively, $\Simp B$ ensures that all hyper-cubes in $B$ are determined uniquely by the composite
of simplices and so while the hyper-cubes of $\Simp B$ may be quite different than in $B$, the
collection of \emph{simplices} in $B$ remains unchanged. The following records this relationship:
\begin{restatable}{axiom}{simplicialstability}
  \label{ax:simp-stability}
  For every $\DeclVar{A}{\GM}{\Uni}$, the following holds:
  \[
    \Prod{n : \Nat}
    \IsEquiv\prn{\prn{\eta_*}^\dagger : \Modify[\GM]{\Delta^n \to A} \to \Modify[\GM]{\Delta^n \to \Simp A}}
  \]
\end{restatable}

Our final axiom in this section records the right adjoint to $\Int \to -$ mentioned in
\cref{sec:intro}. We must be somewhat more careful in stating this than \citet{licata:2018}, as we
wish to ensure that the postulated adjoint is fully coherent, so we use the following formulation
which records that the adjoint exists uniquely:
\begin{restatable}{axiom}{amazingradj}
  \label{ax:amazing-right-adjoint}
  For every $\DeclVar{A}{\GM}{\Uni}$, the following holds:
  \[
    \Sum{\DeclVar{A_\Int}{\GM}{\Uni}, \DeclVar{\epsilon}{\GM}{\prn{A_\Int}^\Int \to A}}
    \Prod{\DeclVar{B}{\GM}{\Uni}} \IsEquiv\prn{\Modify[\GM]{B \to A_\Int} \to \Modify[\GM]{B^\Int \to A}}
  \]
\end{restatable}

Intuitively, $A_\Int$ is the application of this ``amazing'' right adjoint to a
$\GM$-annotated type $A$ and $\epsilon$ is the counit of this adjunction at $A$. From this data we
may reconstruct a functor $\prn{-}_\Int : \Modify[\GM]{\Uni} \to \Modify[\GM]{\Uni}$.
Crucially, this axiom only applies when $A$ is $\GM$-annotated; \citet{licata:2018} show that
requiring it for arbitrary types forces $\Int = \ObjTerm{}$, contradicting
 $0 \neq 1 : \Int$. Moreover, as noted earlier, this axiom is \emph{not} validated by the
model of \STT{} in simplicial spaces; it is only after shifting to cubical spaces that
\cref{ax:amazing-right-adjoint} is valid. Concretely, it is often the case that even if $A$ is known
to be simplicial, the same will not be true of $A_\Int$. Our main use of this axiom is to ``transpose'' various
$\GM$-annotated predicates $X^\Int \to \Prop$ into predicates $X \to \Prop$. For convenience, we
bundle up this process into the following lemma:

\begin{lemma}\label{lem:primer:amazing-transpose}
  If $\DeclVar{\phi}{\GM}{\Uni^{\Int^n} \to \Prop}$, there is a
  $\DeclVar{\bar{\phi}}{\GM}{\Uni \to \Prop}$ equipped with a canonical equivalence:
  \[
    \Prod{\DeclVar{A}{\GM}{X \to \Uni}}
    \Modify[\GM]{\prn{x : X} \to \bar{\phi}\prn{A\,x}}
    \Equiv
    \Modify[\GM]{\prn{x : X^{\Int^n}} \to \phi\prn{A \circ x}}
  \]
\end{lemma}
\begin{proof}
  Fix a predicate $\DeclVar{\phi}{\GM}{\Uni^\Int \to \Prop}$ (for simplicity, we handle only the case of
  $n = 1$ as the case general case is identical modulo notational clutter). We begin by using
  \cref{ax:amazing-right-adjoint} to obtain a map  $\DeclVar{\phi_\Int}{\GM}{\Uni \to \Prop_\Int}$
  where $\Prop_\Int$ is the unique type arising from applying the amazing right adjoint to $\Prop$.
  Next, following \citet[\S 3.3]{gratzer:2024}, we observe that the tautological family
  $\ObjTerm{\Int} \to \Prop_\Int$ is classified by a map $\Prop_\Int \to \Prop$ (in particular, it
  is a small family) and, composing this with $\phi_\Int$, we obtain our desired $\bar{\phi} : \Uni
  \to \Prop$. In total, we have the following diagram:
  \[
    \begin{tikzpicture}[diagram]
      \SpliceDiagramSquare{
        nw = \ObjTerm{\Int},
        ne = \ObjTerm{},
        sw = \Prop_\Int,
        se = \Prop,
        nw/style = pullback,
      }
      \node[left = of sw] (A) {$\Uni$};
      \path[->] (A) edge node[above] {$\phi_\Int$} (sw);

      \path[->, bend right] (A) edge node[below] {$\bar{\phi}$} (se);
    \end{tikzpicture}
  \]

  To show that the desired equivalence holds, fix $\DeclVar{A}{\GM}{X \to \Uni}$. We wish to show
  that $\Prod{x : X} \bar{\phi}\prn{A\,x}$ holds if and only if
  $\Prod{x : X^\Int} \bar{\phi}\prn{A \circ x}$. The former holds if and only if there is a
  (necessarily unique) extension of the above diagram:
  \[
    \begin{tikzpicture}[diagram]
      \SpliceDiagramSquare{
        nw = \ObjTerm{\Int},
        ne = \ObjTerm{},
        sw = \Prop_\Int,
        se = \Prop,
        nw/style = pullback,
      }
      \node[left = of sw] (A) {$\Uni$};
      \path[->] (A) edge node[below] {$\phi_\Int$} (sw);
      \node[left = of nw] (X) {$X$};
      \path[->] (X) edge node[left] {$A$} (A);
      \path[->, densely dotted] (X) edge (nw);
    \end{tikzpicture}    
  \]
  After transposing $\prn{-}_\Int$ (and discarding the right-hand vertical map as it is redundant
  for our purposes), we see that the left-hand triangle of this diagram is precisely equivalent to
  the following:
  \[
    \begin{tikzpicture}[diagram]
      \node (nw) {$\ObjTerm{}$};
      \node[below = 1.5cm of nw] (sw) {$\Prop{}$};
      \path[->] (nw) edge (sw);
      \node[left = 2cm of sw] (A) {$\Uni^\Int$};
      \path[->] (A) edge node[below] {$\phi$} (sw);
      \node[left = of nw] (X) {$X^\Int$};
      \path[->] (X) edge node[left] {$A_*$} (A);
      \path[->, densely dotted] (X) edge (nw);
    \end{tikzpicture}    
  \]
  Examining this diagram yields the desired conclusion.
\end{proof}
Finally, \citet{gratzer:2024} have shown that \TTT{} (\MTT{} with this mode theory
extended with all of the above axioms) has a model in cubical spaces. They further show, based on a
result of \citet{riehl:2017}, that categories in \TTT{} are realized by a standard model of
$\infty$-categories: complete Segal spaces~\citep{rezk:2001}.
\begin{theorem}
  \label{thm:primer:soundness}
  There is a model of \TTT{} in cubical spaces $\PSH[\SSET]{\CUBE}$. In this model, categories
  are realized by complete Segal spaces.
\end{theorem}

\subsection{Category theory in triangulated type theory}

We require some of the category theory developed previously in \TTT{} and
\STT{}~\citep{riehl:2017,buchholtz:2023,gratzer:2024,gratzer:2025}. To keep this paper more
self-contained, we recall the relevant results and definitions here.

As noted by \citet{riehl:2017}, a natural transformation between functors $f,g : C \to D$---\ie{}, an
element of $\Hom{f}{g}$---corresponds precisely to a family $\Prod{c : C} \Hom{f\,c}{g\,c}$.
Consequently, a pointwise invertible natural transformation is invertible. We note a refinement of
this statement by \citet{gratzer:2025} which further reduces this to $\GM$ elements (\ie{}, objects) of
$C$:%
\footnote{\citet{gratzer:2025} prove
  \cref{lem:primer:pw-iso-to-iso,lem:primer:pw-left-adjoint-to-left-adjoint} using the twisted arrow
  modality, which we have chosen not to include in \TTT{} for simplicity. More elementary proofs
  merely relying on \cref{ax:cubes-separate} are possible and so there is no issue with their use in
  \TTT{}.}
\begin{lemma}
  \label{lem:primer:pw-iso-to-iso} If $\DeclVar{C,D}{\GM}{\Uni}$ are categories and
  $\DeclVar{f,g}{\GM}{C \to D}$, then a natural transformation $\DeclVar{\alpha}{\GM}{\Hom{f}{g}}$
  is invertible if and only if for all $\DeclVar{c}{\GM}{C}$ the map $\alpha\prn{c} :
  \Hom{f\,c}{g\,c}$ is invertible.
\end{lemma}

We similarly have a synthetic version of the classical result that full, faithful and essentially
surjective functors are equivalences.
\begin{lemma}
  \label{lem:primer:ff-ess-surj}
  If $\DeclVar{C,D}{\GM}{\Uni[\IsCat]}$, then $\DeclVar{f}{\GM}{C \to D}$ is invertible
  iff $f$ is essentially surjective and fully faithful on $\GM$-elements of $C$.
\end{lemma}

Our calculations with cocartesian fibrations in \cref{sec:cocartesian,sec:cat} will rely on the
theory of adjunctions in \TTT{}. To begin with, an adjunction between two functors $f : C \to D$ and
$g : D \to C$ is given by a collection of equivalences:
\[
  \alpha : \Prod{c : C}\Prod{d : D} \Hom{f\,c}{d} \Equiv \Hom{c}{g\,d}
\]
Note that we do not require any additional naturality constraints on $\alpha$, these are
automatically enforced by virtue of working synthetically. We say $f$ is a left adjoint if there exists
a (necessarily unique) $\prn{g,\alpha}$, and dually that $g$ is a right adjoint if there exists
$\prn{f,\alpha}$. It is often difficult to construct such a family of equivalences directly, so we
often use the following result of \citet{gratzer:2025}:
\begin{lemma}
  \label{lem:primer:pw-left-adjoint-to-left-adjoint}
  If $\DeclVar{C,D}{\GM}{\Uni}$ are categories, then $\DeclVar{f}{\GM}{C \to D}$ is a left adjoint
  iff for all $\DeclVar{d}{\GM}{D}$ there exists $\DeclVar{c}{\GM}{C}$ and
  $\DeclVar{\epsilon}{\GM}{\Hom{f\prn{c}}{d}}$ such that the following is an equivalence for all
  $\DeclVar{c'}{\GM}{C}$:
  \[
    \epsilon_* \circ f : \Hom{c'}{c} \to \Hom{f\prn{c'}}{d}
  \]
\end{lemma}

We shall also have use for the various concrete examples of categories constructed by
\citet{gratzer:2024}. Foremost among these is the category of groupoids $\Space$---the
$\infty$-categorical analog of the category of sets. Like our eventual definition of the
category of categories, this is characterized through a universal property.
\begin{definition}
  A family of types $A : X \to \Uni$ is \emph{covariant} if it is right orthogonal to the inclusion
  $\brc{0} \to \Int$.
\end{definition}
More intuitively, covariant families are families of groupoids such that synthetic homomorphisms of the
base lift coherently to functors of the fibers. We shall give more exposition of this idea indirectly
in \cref{sec:cocartesian} when we study their generalization: cocartesian families. Covariant
families are closed under numerous properties, including precomposition, $\Sum{}$-types, \etc{} We
now define $\Space$ as the base of the universal covariant family:
\begin{definition}
  $\Space$ is the unique subtype of the universe such that $\Space \to \Uni$ is covariant, and for
  all $\DeclVar{X}{\GM}{\Uni}$, the canonical map $\Modify[\GM]{X \to \Space} \to
  \Modify[\GM]{\Sum{A : X \to \Uni} \IsCov\prn{A}}$ is an equivalence.
\end{definition}

Consequently, the type of objects of $\Space$, \ie, $\Modify[\GM]{\Space}$, is immediately seen to be
equivalent to $\GM$-covariant families over $\ObjTerm{}$. These, in turn, are equivalent to
$\GM$-groupoids. The main result of \citet{gratzer:2024} extends this characterization to synthetic
morphisms:
\begin{theorem}
  $\Space$ is a category, and there is a canonical equivalence $\Hom[\Space]{A}{B} \Equiv \prn{A \to
  B}$. Moreover, under this equivalence, identities and composition of synthetic morphisms are
  realized by identity functions and ordinary function composition.
\end{theorem}

Finally, we record a minor but useful result stating that $\Int$,
and some types derived from it, form categories~\citep{gratzer:2024}.
\begin{lemma}
  $\Int$, $\Int^n$, and $\Delta^n$ all form categories.
\end{lemma}

\subsection{Technical results on iso-inner fibrations}

Before moving on to the main focus of this article, we record a few technical results about
iso-inner fibrations and the $\Simp$ modality.

\begin{definition}
  We write $\Spine^n$ for the iterated pushout $\Int \Pushout{\ObjTerm{}} \dots \Pushout{\ObjTerm{}}
  \Int$ which glues together $n$ copies of $\Int$ attaching $0$ to $1$ and $1$ to $0$.
\end{definition}

\begin{lemma}
  \label{lem:proofs:spine-from-segal} If $\DeclVar{f}{}{X \to Y}$ is inner it is orthogonal to
  $\Spine^n \to \Delta^n$ for all $n \ge 2$.
\end{lemma}
\begin{proof}
  The $n = 2$ case is by definition, so we proceed by induction. By induction hypothesis, $\Spine^{n
  + 1} \to \Delta^n \Pushout{\ObjTerm{}} \Int$ is orthogonal to all inner maps(left maps are closed
  under pushouts). It suffices to show that $\Delta^n \Pushout{\ObjTerm{}} \Int \to \Delta^{n + 1}$
  is orthogonal to all inner maps. Unfolding these conditionals, we must show the following:
  \[ \Compr{\prn{v,i}}{v\prn{n} = 1 \lor i = 0} \to \Compr{\prn{v,i}}{v\prn{n} \ge i} \] This, in
  turn, is a retract of $\Int^{n - 1} \times \Lambda^2_1 \to \Int^{n- 1} \times \Delta^2$.
\end{proof}

The converse of this lemma also holds if we further assume that the relevant types are simplicial.
It is an easy corollary of the following result:
\begin{lemma}\label{lem:proofs:segal-from-spine}
  If $\DeclVar{X,Y}{\GM}{\Uni[\Simp]}$ and $\DeclVar{f}{\GM}{X \to Y}$ is
  $\GM$-orthogonal to $\Spine^n \to \Delta^n$ so too is $f^\Int$
\end{lemma}
\begin{proof}
  We must show that $\Spine^n \times \Int \to \Delta^n \times \Int$ is $\GM$-orthogonal
  $f$. To this end, we note the following identification:
  \[
    \Delta^n \times \Int 
    =
    \Compr{
      \prn{v,i} : \Delta^n \times \Int
    }{
      \exists k \in \brc{0, \dots, n+1}.\,v\prn{k} \ge i \ge v\prn{k + 1}
    }
  \]
  Here, by convention, we treat $v\prn{0} = 1$ and $v\prn{n+1} = v\prn{n+2} = 0$. In what follows,
  we write $\Phi_k$ for the condition $v\prn{k} \ge i \ge v\prn{k + 1}$.

  Consequently, to show the desired lifting it suffices to show that there is a unique such lift for
  $\Compr{\prn{v,i} : \Delta^n \times \Int}{\Phi_{k_0} \times \dots \times \Phi_{k_l}}$. A moment's
  thought reveals that each such intersection is a sub-simplex of $\Delta^n \times \Int$. In
  particular, each $\Compr{\prn{v,i}}{\Phi_k}$ is $\Delta^{n+1}$ and each higher intersection is a
  smaller simplex. Moreover, its intersection with $\Spine^n \times \Int$ is either exactly the spine of this smaller
  simplex (in the cases of $\Phi_0$ and $\Phi_n$ or higher intersections) or
  $\Int \Pushout{\ObjTerm{}} \dots \Pushout{\ObjTerm{}} \Delta^2 \Pushout{\ObjTerm{}} \dots$.
  In either case, the unique lifting exists by assumption on $f$ (in the latter case, by 2-for-3 and
  the closure of left classes under pushouts, as one may see by observing the following
  decomposition:
  $\Spine^{n + 1} 
  \to
  \prn{
    \Int \Pushout{\ObjTerm{}} \dots
    \Pushout{\ObjTerm{}} \Delta^2 \Pushout{\ObjTerm{}}
    \dots \Pushout{\ObjTerm{}} \Int
  }
  \to
  \Delta^{n + 1}$).
\end{proof}

\begin{corollary}
  \label{cor:proofs:spines}
  If $\DeclVar{X}{\GM}{\Uni}$ is simplicial, it suffices to show that it is $\GM$-orthogonal to
  $\Spine^n \to \Delta^n$ to prove that it is Segal.
\end{corollary}

\begin{lemma}
  \label{lem:proofs:simp-preserves-segal}
  If $\DeclVar{X}{\GM}{\Uni}$ is Segal so too is $\Simp X$.
\end{lemma}
\begin{proof}
  To prove this, we argue that $\Simp X$ is $\GM$-orthogonal\footnote{Meaning, orthogonal when
  we restrict our attention to $\GM$-annotated maps} to $\Spine^n \to \Delta^n$.
  We note that this follows immediately from simplicial stability (\cref{ax:simp-stability}) along with
  the fact that $\Modify[\GM]{-}$ commutes with limits:
  \begin{align*}
    & \Modify[\GM]{\Spine^n \to \Simp X}
    \\
    &\Equiv \Modify[\GM]{\Int \to \Simp X}
    \times_{\Modify[\GM]{\Simp X}} \Modify[\GM]{\Int \to \Simp X}
    \times_{\Modify[\GM]{\Simp X}} \dots
    \\
    &\Equiv \Modify[\GM]{\Int \to X} \times_{\Modify[\GM]{X}} \Modify[\GM]{\Int \to X} \times_{\Modify[\GM]{X}} \dots
    \\
    &\Equiv \Modify[\GM]{\Spine^n \to X}
  \end{align*}
  Consequently, we are reduced to the same question for $X \to 1$ which follows from
  \cref{lem:proofs:spine-from-segal}.
\end{proof}

\begin{lemma}
  \label{lem:proofs:simp-preserves-rezk}
  If $\DeclVar{X}{\GM}{\Uni}$ is Segal and Rezk then $\Simp X$ is Rezk.
\end{lemma}
\begin{proof}
  We must show that $\Simp X$ is $\GM$-orthogonal to $\IntIso \times \Delta^n \to \Delta^n$ and, since
  $\Simp X$ is Segal, we may reduce immediately to the case where $n = 0$ or $n = 1$. The first case
  is an immediate consequence of simplicial stability, as $\IntIso$ is built by pushing out various
  simplices and therefore maps $\Modify[\GM]{\IntIso \to \Simp X}$ correspond to those maps which factor
  through $X$.

  For the $n = 1$ case, we must show that diagrams of the following shape in $X$ are determined by
  the bottom-most edge

  \begin{equation*}
    \begin{tikzpicture}[diagram]
      \node (A) {$x_0$};
      \node [right = of A] (B) {$x_1$};
      \node [below = 1.5cm of A] (C) {$x_2$};
      \node [right = of C] (D) {$x_3$};
      \path[->] (A) edge node [above] {$f$} (B);
      \path[->] (A) edge node[left] {$\iota_1$\rotatebox[origin=c]{-90}{$\sim$}} (C);
      \path[->] (B) edge node[right] {\rotatebox[origin=c]{-90}{$\sim$} $\iota_0$} (D);
      \path[->] (C) edge node[below] {$g$} (D);
      \path[->] (A) edge (D);
    \end{tikzpicture}
  \end{equation*}

  By simplicial stability along with the previous case, we may safely assume that each of these
  components (including the relevant 2-cells and section-retraction pairs) all come from $X$.
  Moreover, since $X$ is Segal, the top simplex is redundant. In particular, it is equivalent to the
  type
  $\Sum{\DeclVar{\iota_1}{\GM}{\IntIso \to X}}
  \Sum{\DeclVar{f}{\GM}{\Int \to X}}
  \Sum{\DeclVar{p}{\GM}{\iota\prn{1} = f\prn{0}}}
  \iota \circ_p f = g \circ \iota_1$.
  By the $n = 0$ case we have already discussed, we may assume $\iota_1 = \ArrId{}$ and replace
  $\DeclVar{\iota_1}{\GM}{\IntIso \to X}$ by simply $\DeclVar{x}{\GM}{X}$. After this replacement, the
  whole type collapses to a singleton type.

  It finally suffices to show that the bottom triangle is equivalent to the bottom edge. However, we
  may replace this triangle by the corresponding inner horn, whereafter another application of the
  Rezk condition for $n = 0$ finishes things off.
\end{proof}

\begin{corollary}
  \label{cor:proofs:simp-cat}
  If $\DeclVar{X}{\GM}{\Uni}$ is Segal and Rezk, then $\Simp X$ is a category.
\end{corollary}

\section{Recollections on cocartesian fibrations}
\label{sec:cocartesian}

Our eventual goal of characterizing $\Modify[\GM]{X \to \Cat}$ crucially depends on the theory
of \emph{cocartesian families}~\citep{joyal:2008,lurie:2009}. These are a subset of type families $A : X \to \Uni$ for which
(1) each $A\,x$ is a category, and (2) each morphism $f : \Hom{x}{y}$ in $X$ induces a transport
function $A\,x \to A\,y$. We shall require that these transport functions are
functorial, and that the coherences enforcing functoriality are themselves coherent, \etc{}
In order to structure all of this data, we ask for $A$ to satisfy a number of propositions
that, when combined, give rise to (1) and (2) above. While somewhat indirect, this
accounts for the infinite hierarchy of coherences that would otherwise be impossible to write
down. This material has been developed within \STT{}~\citep{buchholtz:2023}. We recall it
for the reader's benefit.

\subsection{The definition of cocartesian families}
Consider a family $A : X \to \Uni$ and write $\rho^A$ for the canonical map given by restriction and
projection
$\TotalTy{A}^{\Delta^2} \to \TotalTy{A}^{\Lambda^2_0} \times_{X^{\Lambda^2_0}} X^{\Delta^2}$.
Note that the restriction $\TotalTy{A}^{\Delta^2}$ to $\TotalTy{A}^{\brc{\bar{0} \to \bar{1}}}$
along with the corresponding map
$\TotalTy{A}^{\Lambda^2_0} \times_{{\Lambda^2_0} \to X} X^{\Delta^2} \to \TotalTy{A}^{\brc{\bar{0} \to \bar{1}}}$
commute with $\rho^A$, so we may view $\rho^A$ as a map between two families of types over
$\Int \to \TotalTy{A}$. Given $x : \Int \to \TotalTy{A}$, we denote by $\rho^A_x$ the
restriction of $\rho^A$ to the fibers of these families over $x$.

\begin{definition}\label{def:cocartesian:arrow}
  A morphism $f : \Int \to \TotalTy{A}$ is \emph{cocartesian} in $A$
  (written $\Con{isCocartArr}_A\prn{f}$) if $\rho^A_f$ is an equivalence.
\end{definition}

Informally, $f : \Int \to \TotalTy{A}$ is cocartesian if diagrams of the following shape have a
unique lift whenever $\Delta^{\brc{\bar{0} \to \bar{1}}} \to \Delta^2 \to \TotalTy{A}$ is $f$:
\[
  \DiagramSquare{
    height = 1.2cm,
    nw = \Lambda^2_0,
    sw = \Delta^2,
    ne = \TotalTy{A},
    se = X,
  }
\]

\begin{definition}
  A family $\DeclVar{A}{}{X \to \Uni}$ is
  \emph{cocartesian} if it is iso-inner, each $A\prn{x}$ is simplicial, and
  the following holds:
  \begin{align*}
    \Prod{\DeclVar{u}{}{\Int \to X}}
    \Prod{\DeclVar{a}{}{A(u\,0)}}
    \Sum{\DeclVar{f}{}{\hom_u^A(a,\bullet)}}
    \mathsf{isCocartArr}_A(f)
  \end{align*}
\end{definition}
\noindent
This requirement is a proposition~\citep{buchholtz:2023} and each fiber of such a family is a
category, satisfying the first of our requirements.
\begin{example}
  For every category $A$ the codomain map ${A^\Int \to A}$ is cocartesian. The domain projection is
  cocartesian iff $A$ has pushouts.
\end{example}
\begin{remark}
  An arrow is \emph{vertical} if it maps to an isomorphism. In a cocartesian fibration, every arrow
  factors as ``vertical $\circ$ cocartesian''.
\end{remark}

In case that $X$ is a category and $A$ is simplicial and iso-inner, a slick characterization of
cocartesian families becomes available.
\begin{proposition}[\citet{buchholtz:2023}]
  \label{prop:cocartesian:lari}
  An iso-inner family $\DeclVar{A}{}{X \to \Uni[\Simp]}$ over a category $X$ is cocartesian iff the map
  $(\TotalTy{A})^\Int \to (\TotalTy{A})^{\{0\}} \times_{X^{\{0\}}} X^\Int$ has a left adjoint
  right inverse.\footnote{That is, a left adjoint where the unit map is an isomorphism.}
\end{proposition}

This characterization implies many closure properties such as under composition, pullback, and Leibniz
cotensors~\citep{buchholtz:2023}.

For the second desideratum of cocartesian families, we define the cocartesian transport
operation, providing the desired functors between fibers. If $A : X \to \Uni$ is cocartesian and
$u : \Hom{x}{y}$, transport $\DeclVar{u_!}{}{A\,x \to A\,y}$ is defined by mapping
$\DeclVar{a}{}{A\,x}$ to the codomain of the (unique) cocartesian lift of $u$ starting at $a$.
\begin{proposition}
  If $\DeclVar{A}{}{X \to \Uni}$ is cocartesian, then cocartesian transport is \emph{functorial},
  \ie{}, $(vu)_! = v_! \circ u_!$ and $(\ArrId{})_! = \ArrId{}$.
\end{proposition}

\begin{definition}
  For $\DeclVar{A,B}{}{X \to \Uni}$ cocartesian, we say that
  $f : \Prod{x : X} A\,x \to B\,x$ is a \emph{cocartesian functor} if $f$ preserves cocartesian arrows. We write $A
  \CoCartTo B$ for the type of cocartesian functors.
\end{definition}

For our construction of $\Cat$, it will be helpful to develop the theory of \emph{locally}
cocartesian fibrations. These are families $A : X \to \Uni$ that are cocartesian after
restriction $A \circ f$ for every $f : \Int \to X$. As setup,
for a family $A : \Int \to \Uni$, we call an edge $a : \prn{i : \Int} \to A\,i$
\emph{locally cocartesian} if the following proposition holds:
\begin{align*}
  &\Con{isLocallyCoCart} : \prn{A : \Int \to \Uni} \to \prn{\prn{i : \Int} \to A\,i} \to \Prop
  \\
  &\Con{isLocallyCoCart}\,A\,a =
  \Prod{b : \prn{i : \Int} \to A\,i} \Prod{p : a\,0 = b\,0}
  \\
  &\quad \IsContr\prn{\Sum{t : \prn{i,j : \Delta^2} \to A\,i} \Restrict{t}{\Lambda^2_0} = \brk{a,b,p}}
\end{align*}
We then define the structure of having locally cocartesian lifts:
\begin{align*}
  &\Con{hasLCCLifts} : \prn{\Int \to \Uni} \to \Uni
  \\ 
  &\Con{hasLCCLifts}\,A =
  \Prod{a_0 : A\,0} \Sum{a : \Hom[A]{a_0}{\bullet}} \Con{isLocallyCoCart}\,a
\end{align*}

Unlike for cocartesian edges, locally cocartesian edges need not compose. Let us quickly
isolate what it means for them to do so:
\begin{align*}
  &\Con{LCCLiftsCompose} : \prn{\Delta^2 \to \Uni} \to \Prop
  \\
  &\Con{LCCLiftsCompose}\,A = \prn{a : \Prod{s : \Delta^2} A\,s}
  \\
  &\quad \to \Con{isLocallyCoCart}\prn{a\prn{-,0}} \times \Con{isLocallyCoCart}\prn{a\prn{1,-}}
  \\
  &\quad \to \Con{isLocallyCoCart}\prn{\lambda i.\,a\prn{i,i}}
\end{align*}

We extend the preceding two definitions to general families $A : X \to \Uni$ by stating
that they hold for $A$ if they hold for $A \circ f : \Int \to \Uni$ for all arrows $f :
\Int \to X$ (or squares $h : \Int \times \Int \to X$, respectively). We overload the predicates, writing, \eg{},
$\Con{hasLCCLifts}\prn{A}$ also for a general family $A$. A \emph{locally cocartesian family} is one with
$\Con{hasLCCLifts}$ structure.

\begin{lemma}
  \label{lem:proofs:lcc-char}
  If $A : X \to \Uni[\Simp]$ is an iso-inner fibration, then a dependent edge $a : \prn{i :
  \Int} \to A\prn{x\,i}$ over $x : \Int \to X$ is locally cocartesian if and only if the following
  map is an equivalence:
  \[
    a^* :
    \prn{\Sum{f : \Int \to A\prn{x\,1}} f\prn{0} = a\prn{1}}
    \to
    \prn{\Sum{f : \prn{i : \Int} \to A\prn{x\,i}} f\prn{0} = a\prn{0}}
  \]
\end{lemma}
\begin{proof}
  This question is once more restricted to a particular $x : \Int \to X$, we may pull back $A$ along
  this map to suppose that $A : \Int \to \Uni[\Simp]$ and that $a : \prn{i : \Int} \to A\,i$. In
  this case, $a$ is locally cocartesian if and only if the following holds (by definition):
  \[
    \Prod{b : \prn{i : \Int} \to A\,i} \Prod{p : a\,0 = b\,0}
    \IsContr\prn{\Sum{t : \prn{i,j : \Delta^2} \to A\,i} \Restrict{t}{\Lambda^2_0} = \brk{a,b,p}}
  \]

  Fix $b : \prn{i : \Int} \to A\,i$ and $p : a\,0 = b\,0$. Then 
  $\Sum{t : \prn{i,j : \Delta^2} \to A\,i} \Restrict{t}{\Lambda^2_0} = \brk{a,b,p}$ is equivalent
  (by innerness) to $c : \Int \to A\,1$ along with
  $q : c\prn{0} = a\prn{1}$, $\theta : c \circ_q a = b$, and $\theta_0 : \Ap_{-\prn{0}}\prn{\theta}
  = p$. However, this is precisely the fiber of $a^*$, so an edge is locally cocartesian if and only
  if each fiber of $a^*$ is contractible.
\end{proof}

\begin{theorem}\label{thm:cocartesian:lcc-vs-cc}
  If $A : X \to \Uni[\Simp]$ is iso-inner and locally cocartesian where locally cocartesian edges
  compose, then locally cocartesian edges are cocartesian and $A$ is cocartesian.
\end{theorem}
\begin{proof}
  We wish to show that $\Lambda^2_0 \to \Delta^2$ is orthogonal to $\tilde{A} \to X$ if the
  $0 \to 1$ edge of $\Lambda^2_0$ is sent to a locally cocartesian edge. Since this property can be
  tested after pulling back $\Sum{x : X} A\,x \to X$, we may assume that $X = \Delta^2$ and concern
  ourselves only with the tautological 2-simplex. Moreover, in this situation $\Sum{x : X} A\,x$ is
  a (simplicial) category.

  Let us now fix $\brk{f,g,p} : \Lambda^2_0 \to \Sum{x : X} A\,x$ which lifts $\ArrId{}$ such that $f$
  is locally cocartesian. Let us write $x = f\prn{0}$, $y = f\prn{1}$, and $z = g\prn{1}$. With this
  notation, $p : x = g\prn{0}$.

  We wish to show that $\brk{f,g,p}$ extends uniquely to a 2-simplex (such a 2-simplex will
  necessarily lie correctly over the unique non-degenerate 2-simplex in $X$ and---since $X$ is a
  set---there is no interesting data in how it lies over this simplex). To this end, it suffices to
  show that the following map is an equivalence~\cite[Proposition 5.1.10]{buchholtz:2023}:
  \[
    f^* : \Hom{y}{z} \to \Hom{x}{z} 
  \]

  If so, we may choose the (unique) preimage of $\prn{g,p,\Refl}$ to conclude the proof. Now, let
  us choose a locally cocartesian lift of $\prn{1,-} : \Int \to \Delta^2$ with starting point $y$.
  This is a morphism $h'$. Let us write $w$ for the target of $h'$. Since $h' \circ f$ is locally
  cocartesian by assumption, we conclude that the following maps are equivalences by
  \cref{lem:proofs:lcc-char}:
  \[
    (h' \circ f)^* : \Hom{w}{z} \to \Hom{x}{z} 
    \qquad
    h'^* : \Hom{w}{z} \to \Hom{y}{z} 
  \]
  By 3-for-2, the same is then true of $f^* : \Hom{y}{z} \to \Hom{x}{z}$ as required.
\end{proof}
In this case, locally cocartesian lifts are unique, since cocartesian lifts are.
This will be important shortly.

\begin{lemma}\label{lem:cat:lcc-prop}
  If $A : X \to \Uni[\Simp]$ is iso-inner, then $\Con{hasLCCLifts}\prn{A}$ and
  $\Con{LCCLiftsCompose}\prn{A}$ are propositions.
\end{lemma}
\begin{proof}
  Note that it suffices to prove that for all $x : \Int \to X$ (respectively, $x : \Int^2 \to X$)
  that $\Con{hasLCCLifts}\prn{A \circ x}$ (respectively, $\Con{LCCLiftsCompose}\prn{A \circ x}$) is
  a proposition.

  Only the first of these obligations is non-trivial, since $\Con{isLocallyCoCart}$ is manifestly
  valued propositions. Note that if $a$ and $a'$ are both locally cocartesian lifts, then by
  construction there is a unique vertical isomorphism $\iota : a\prn{1} \cong a'\prn{1}$ such that
  $\iota \circ a = a'$. We may then consider these as isomorphic arrows in the fiber
  $\Int \times_{X} \Sum{x : X} A\,x$, which is a category by iso-innerness. Consequently, $a = a'$
  in the fiber $\Int \times_{X} \Sum{x : X} A\,x$ whereby they are equal in $\Sum{x : X} A\,x$.
  The type of locally cocartesian lifts is therefore a proposition as required.
\end{proof}

\subsection{The directed gluing of cocartesian families}
We close this section by generalizing the \emph{directed gluing type} inspired by
\citet{weaver:2020} and used in this context by \citet{gratzer:2024}. Roughly, this type takes
two cocartesian families over $X$ and a cocartesian functor between them and bundles them into a single
cocartesian family over $X \times \Int$. This is a key ingredient of our proof of directed
univalence, which will eventually amount to a proof that $\Glue$ lifts to an equivalence.

Fix cocartesians fibrations $\DeclVar{F_0,F_1}{}{X \to \Uni[\Simp]}$ and a cocartesian functor
$\DeclVar{\alpha}{}{\Prod{x : X} F_0\,x \to F_1\,x}$. The \emph{directed gluing} of this data is

\iblock{
  \mrow{\Glue\prn{F_0,F_1,\alpha} : X \times \Int \to \Uni[\Simp]}
  \mrow{
    \Glue\prn{F_0,F_1,\alpha}\prn{x,i} =
    \Sum{f : F_1\prn{x}} i = 0 \to \alpha\prn{x}^{-1}\prn{f}
  }
}

We note that the fibers over $\prn{x,0}$ and $\prn{x,1}$ are given by $F_0\prn{x}$ and $F_1\prn{x}$,
respectively.  Moreover, for each $w : F_0\prn{x}$ there is a map over $\lambda i.\,\prn{x,i}$
connecting $w$ to $\alpha\prn{x,w}$. We show that $\Glue\prn{F_0,F_1,\alpha}$ is iso-inner over
$\Int \times X$ and that the aforementioned collection of edges make this family cocartesian
with transport functor $\alpha$.

In what follows, we assume that $X,F_0,F_1$ and $\alpha$ are all $\GM$-annotated.
These proofs are all routine applications of orthogonality properties, combined with
\cref{lem:primer:pw-left-adjoint-to-left-adjoint}.

\begin{lemma}\label{lem:cocartesian:glue-is-isoinner}
  If $X$ is simplicial, then $\Glue\prn{F_0,F_1,\alpha}$ is iso-inner.
\end{lemma}
\begin{proof}
  First, we note by elementary manipulation of closure properties, it suffices to consider the case
  where $F_1 = \lambda \_.\,\ObjTerm{}$
  (use the 3-for-2 fact available for iso-inner families with the factorization
  $\Glue\prn{F_0,F_1,\alpha} \to \prn{\Sum{x : X} F_1\prn{x}} \times \Int \to X \times \Int$).

  We must show that $\Spine^n \to \Delta^n$ is $\GM$-orthogonal to $\Glue\prn{F_0,F_1,\alpha}$ by
  \cref{cor:proofs:spines}. To this end, fix a map $\DeclVar{b}{\GM}{\Delta^n \to X \times \Int}$
  along with a partial section:
  \[
    \DeclVar{s_0}{\GM}{\prn{v : \Spine^n} \to \Glue\prn{F_0,F_1,\alpha}\prn{b\prn{v}}}
  \]
  We must show that $s_0$ extends uniquely. Let us begin by investigating $i = \Proj[1] \circ b$, a
  $\GM$-annotated map $\Delta^n \to \Int$. By duality, such a map corresponds to a $\GM$ element of
  $\Int\brk{x_0 \le \dots \le x_n}$ so it is either $0$, $1$, or $x_k$ for some particular $0 \le k
  \le n$. If we are in the case of $i = \lambda \_.\,0$ or $i = \lambda \_.\,1$, then the conclusion
  is immediate from the innerness of $F_0$. If we are instead in the case where $i =
  \Con{dual}\prn{x_k}$, we must proceed differently. We note that in such a case, $i\prn{v} = 0$ if
  and only if $v_k = 0$. This means that $s_0$ has the type  $\prn{v : \Spine^{k}} \to
  F_0\prn{b\prn{v, 0 \dots}}$ and we must construct a unique extension $s$ of type $\prn{v :
  \Delta^{k}} \to F_0\prn{b\prn{v, 0, \dots}}$. This is immediate by the innerness of
  $F_0\prn{b\prn{-, 0 \dots}}$ which in turn is a consequence of the innerness of $F$.

  To show the ``iso'' part of iso-innerness, we note that this can be checked fiberwise. However,
  over each fiber, this is immediate by the fact that Rezk types are an exponential ideal and $F_0$
  is iso-inner.
\end{proof}

\begin{lemma}\label{thm:cocartesian:glue-is-cocart}
  If $X$ is a category, then $\Glue\prn{F_0,F_1,\alpha}$ is cocartesian.
\end{lemma}
\begin{proof}
  In light of \cref{lem:cocartesian:glue-is-isoinner}, everything involved is a (simplicial) category.  
  Accordingly, we may use the LARI condition of \cite{buchholtz:2023} to prove this result.
  Applying the results of \cite{gratzer:2025}, we are then further reduced to constructing this
  adjoint on objects.

  Fix $\DeclVar{i}{\GM}{\Int \to \Int}$ along with $\DeclVar{x}{\GM}{\Int \to X}$,
  $\DeclVar{f^1_0}{\GM}{F_1\prn{x}}$ and
  $\DeclVar{f^0_0}{\GM}{i\prn{0} = 0 \to \alpha^{-1}_x\prn{f^1_0}}$. We begin by constructing a lift of
  $\prn{f^1_0,f^0_0}$ and then we argue that it is suitably initial.

  First, let us note that since $\alpha$ preserves cocartesian edges, if we have $g : F_0\prn{x}$
  which lies over $f : F_1\prn{x}$, then $x_!g_0$ lies over $x_!f_0$ by a contractible choice of
  path since $\alpha\prn{x_!g_0}$ is a cocartesian lift of the same data as $x_!f_0$. Consequently,
  we may construct the desired lifts:
  \[
    f^1 = \lambda j : \Int.\,x\prn{- \land j}_!\prn{f^1_0}
    \quad
    f^0 = \lambda j : \Int, z : i\prn{j} = 0.\,x\prn{- \land j}_!\prn{f^0_0\prn{\_}}
  \]

  Assume now we are given $\alpha : \Int \times \Int \to \Int$ and $\chi : \Int \times \Int \to X$
  such that $i = \alpha\prn{0,-}$ and $\chi\prn{0,-} = x$ (we silently replace $x$ and $i$ by
  transport along these paths to treat them as reflexivity in what follows). We also assume we are
  given partial lifts:
  $g^1 : \prn{j : \Int} \to F_1\prn{\chi\prn{1,j}}$
  $g^0 : \prn{j : \Int} \to \alpha\prn{1,j} = 0 \to \alpha_{\chi\prn{1,j}}^{-1}\prn{g^1\prn{j}}$
  and
  $h^1_0 : \prn{k : \Int} \to F_1\prn{\chi\prn{k,0}}$.
  $h^0_0 : \prn{k : \Int} \to \alpha\prn{k,0} = 0 \to \alpha^{-1}_{\chi\prn{k,0}}\prn{h^1_0\prn{k}}$.
  We may further assume that all of these are $\GM$-annotated and that there are paths
  $q : \prn{g^1\prn{0},g^0\prn{0}} = \prn{h^1_0\prn{1},h^0_0\prn{1}}$
  and
  $p : \prn{f^1\prn{0},f^0\prn{0}} = \prn{h^1_0\prn{0},h^0_0\prn{0}}$

  We wish to show that there is a unique extension $\prn{h^1_0,h^0_0}$ to all of $\Int \times \Int$
  which matches with $h_0$, $e$, and $f$ over $p$ and $q$. First, we note that we may uniquely
  extend $h^1_0$ to $h^1$ since $F_1$ is cocartesian. Let us therefore replace $h^1_0$, $f^1$, and
  $g^1$ with $h^1$ so that our new goal is to construct an extension of $h^0_0$ given the following:
  \begin{gather*}
    h^0_0 : \prn{k : \Int} \to \alpha\prn{k,0} = 0 \to \alpha^{-1}_{\chi\prn{k,0}}\prn{h^1\prn{k,0}}
    \\
    f^0 : \prn{j : \Int} \to \alpha\prn{0,j} = 0 \to \alpha^{-1}_{\chi\prn{0,j}}\prn{h^1\prn{k,0}}
    \\
    g^0 : \prn{j : \Int} \to \alpha\prn{1,j} = 0 \to \alpha^{-1}_{\chi\prn{1,j}}\prn{h^1\prn{k,1}}
  \end{gather*}

  To prove this, we must perform a somewhat lengthy case analysis on $\alpha$. Since it is
  $\GM$-annotated, we know that it is a $\GM$-element of $\Int\brk{x_0,x_1}$ and we can analyze it
  somewhat extensively.

  \begin{description}
    \item[Case.] $\alpha\prn{0,-} = \lambda \_.1$.
      \leavevmode

      In this case, $\alpha = \lambda \_.1$ by monotonicity, and so any extension is necessarily
      trivial.
    \item[Case.] $\alpha\prn{0,-} = \lambda \_.0$.
      \leavevmode

      Here, we have several sub-cases to consider:
      \begin{description}
        \item[Case.] $\alpha\prn{1,-} = \lambda \_.0$.
          \leavevmode

          In this case, the condition $\alpha\prn{-,-} = 0$ holds in all cases, so this reduces
          precisely to the fact that $F_0$ is cocartesian. In particular, we note that we may extend
          $h^0$ in $F_0$ (not in the fiber) using the fact that $f^0$ is cocartesian. This extension
          is unique by construction and, since the input to the extension lies over $h^1$, it lives
          in the correct fiber (uniquely).

        \item[Case.] $\alpha\prn{1,-} = \lambda j.\,j$.
          \leavevmode

          In this case, we must construct a lift of the following type:
          \[
            \prn{k,j : \Int} \to k \land j = 0 \to \alpha^{-1}_{\chi\prn{k,j}}\prn{h^1\prn{k,j}}
          \]
          Given $k,j : \Int$, since everything is simplicial we may assume that $k \le j$ or $j \le
          k$. In other words, our condition is equivalent to $k = 0$ or $j = 0$; any extension is
          fully determined by the boundary conditions.

        \item[Case.] $\alpha\prn{1,-} = \lambda \_.1$.
          \leavevmode

          In this case, $\alpha = \lambda \prn{k,j}.\, k$ and so we may just take $h_0 = h^0_0$.
      \end{description}

    \item[Case.] $\alpha\prn{0,-} = \lambda j\,.j$.
      \leavevmode

      Here, we have several sub-cases to consider:
      \begin{description}
        \item[Case.] $\alpha\prn{1,-} = \lambda j.\,j$.
          \leavevmode

          In this case, $\alpha\prn{k,j} = j$ and so we may simply take $h^0 = f^0$.

        \item[Case.] $\alpha\prn{1,-} = \lambda \_.1$.
          \leavevmode

          In this case, $\alpha\prn{k,j} = k \lor j$, so our condition $\alpha\prn{k,j} = 0$ amounts
          to $k = 0 \land j = 0$. Consequently, we may take $h = f^0_0$. \qedhere
      \end{description}
  \end{description}
\end{proof}

\begin{corollary} \label{cor:cocartesian:transport}
  Cocartesian transport from
  $\Glue\prn{F_0,F_1,\alpha}\prn{-,0}$ to $\Glue\prn{F_0,F_1,\alpha}\prn{-,1}$
  is given by $\alpha$.
\end{corollary}
\begin{proof}
  First, we note that given $f : \Glue\prn{F_0,F_1,\alpha}\prn{c,0} \cong F_0\prn{c}$, there is a
  functorial choice of edges:
  \[
    \lambda i.\, \prn{\alpha\prn{f}, \lambda \_.\,\prn{f,\Refl}} :
    \prn{i : \Int} \to \Glue\prn{F_0,F_1,\alpha}\prn{c,i}
  \]
  To show the desired identification, it suffices to show that this edge is cocartesian and, using
  the standard result that a natural transformation is an equivalence if and only if it is pointwise
  such, we restrict our attention to the case where $\DeclVar{c}{\GM}{C}$. This, however, is
  immediate in light of the above proof---in particular, cocartesian transport along a constant edge
  in $C$ is trivial.
\end{proof}
\begin{corollary}
  \label{cor:cocartesian:proj-cocart}
  The projection map $\Proj[0] : \Glue\prn{F_0,F_1,\alpha} \to F_1 \circ \Proj[0]$ is a cocartesian
  functor over $X \times \Int$.
\end{corollary}

\section{The universe of amazingly cocartesian types}
\label{sec:cat}

We now turn to the construction of $\Cat$. As mentioned in
the introduction, $\Cat$ will be a subtype of $\Uni$ and therefore must be classified by a proposition
$\Uni \to \Prop$. The most obvious choice of proposition is something akin to being
cocartesian,
but a moment's thought reveals this is unworkable: if we are to define a map $\Con{isCocartFib} :
\Uni \to \Prop$, what should the input be cocartesian over? Cocartesianness is a property of
families!

To fix this, we follow \citet{licata:2018} as refined in the context of directed type
theories~\citep{weaver:2020,gratzer:2024}. First, consider a general notion of fibration
$\DeclVar{\Con{isFib}_X}{}{\Uni^X \to \Prop}$. The goal is to define a predicate that
witnesses fibrancy of a type $A$ viewed as a family over the entire ambient context. From
$\DeclVar{\Con{isFib}_\Int}{\GM}{\Uni^\Int \to \Prop}$, \cref{lem:primer:amazing-transpose} yields
precisely such a notion of fibrancy. Moreover, this stronger notion of fibration can be shown to
agree with the classical notion when we restrict attention to $\GM$-annotated families.
We will now apply this construction, leveraging \cref{thm:cocartesian:lcc-vs-cc}.

Let us write $i : \Delta^2 \to \Int^2$ for the canonical inclusion.
Using \cref{lem:primer:amazing-transpose,lem:cat:lcc-prop}, we now transpose $\Con{isInner}\prn{-
\circ i}$,
$\Con{hasLCCLifts}$, and $\Con{LCCLiftsCompose}\prn{- \circ i}$ to obtain elements of $\Uni \to
\Uni$, namely $\Con{aisInner}$, $\Con{aHasLCCLifts}$, and $\Con{aLCCLiftsCompose}$. Here we have
used $i$ as, \eg{}, $\Con{isInner}$ takes $\Delta^2 \to \Uni$ not $\Int^2 \to \Uni$.
We then define $\Cat$:
\[
  \Cat \defeq
  \Sum{A : \Uni[\Simp]}
  \begin{array}{l}
    \Con{isRezk}\,A \times \Con{aIsInner}\,A 
    \\
    {} \times \Con{aHasLCCLifts}\,A \times \Con{aLCCLiftsCompose}\,A
  \end{array}
\]

\begin{lemma}
  $\Cat$ is a subtype of $\Uni[\Simp]$.
\end{lemma}

\univcocart

\begin{proof}
  Fix $\DeclVar{A}{\GM}{X \to \Uni[\Simp]}$. Our goal is to show that $A$ factors through $\Cat$ if
  and only if $A$ is cocartesian. To prove this, let us consider the data involved in a
  factorization through $\Cat$. By definition, this is equivalent to factoring through four distinct
  subobjects of $\Uni[\Simp]$: those carved out by $\Con{isRezk}$, $\Con{aIsInner}$,
  $\Con{aHasLCCLifts}$, and $\Con{aLCCLiftsCompose}$. In the latter three cases, we may analyze
  these subobjects using the transpositions used to define them.

  For instance, if $A$ factors through $\Sum{B : \Uni[\Simp]} \Con{aIsInner}\prn{B}$, then there is
  an element of the following type by \cref{lem:primer:amazing-transpose}:
  \[
    \Modify[\GM]{
      \Prod{x : X^{\Int \times \Int}}
      \Con{isInner}\prn{A \circ x \circ i}
    }
  \]
  Such an element exists if and only if $A$ is an inner fibration.

  This reasoning applies to $\Con{aHasLCCLifts}$ and $\Con{aLCCLiftsCompose}$ so we may conclude
  that $A$ factors through $\Cat$ if and only if it is iso-inner, locally cocartesian, and locally
  cocartesian edges compose. The desired bi-implication is then
  \cref{thm:cocartesian:lcc-vs-cc}.
\end{proof}

\section{The category of categories}
\label{sec:dua}

In this section, we leverage \cref{thm:cat:univ-cocart} to prove the crucial properties of $\Cat$.
Namely, we prove $\Cat$ is Segal and Rezk, satisfies directed univalence as described in
\cref{sec:intro}, and is simplicial. Combining these results together we show that $\Cat$ is a
category and, in particular, the category of categories.

\subsection{Classifying cocartesian fibrations}
The main input to the proofs that $\Cat$ is Segal and Rezk
is a characterization of cocartesian fibrations over
$\Delta^n \times C$ where $C$ is a category. To see why, note that by
\cref{thm:cat:univ-cocart}, we know that $\DeclVar{f}{\GM}{X \times \Delta^n \to \Cat}$ is
determined by a cocartesian family over $X \times \Delta^n$. By giving a precise description of
such families, we obtain a more tractable version of, \eg{}, the restriction map
$\Modify[\GM]{\Int^n \times \Delta^2 \to \Cat} \to \Modify[\GM]{\Int^n \times \Lambda^2_1 \to \Cat}$.
This version will be manifestly invertible, and so we can conclude that $\Cat$ is Segal.
A key lemma in this process is the following:
\begin{lemma}
  \label{lem:dua:fw-equiv}
  For $\DeclVar{X}{\GM}{\Uni}$ a category and $\DeclVar{A,B}{\GM}{X \to \Uni}$ cocartesian,
  a cocartesian functor $\alpha : \Prod{x : X} A\prn{x} \to B\prn{x}$ induces an equivalence of
  total categories $\TotalTy{\alpha} : \TotalTy{A} \Equiv \TotalTy{B}$ iff\/
  $\Prod{\DeclVar{x}{\GM}{X}} \IsEquiv\prn{\alpha\prn{x}}$ holds.
\end{lemma}
\begin{proof}
  Since cocartesian families are isofibrations, we know that $\TotalTy{A}$ and
  $\TotalTy{B}$ are both categories themselves. By \cref{lem:primer:ff-ess-surj}, we 
  check that $\TotalTy{\alpha}$ is fully faithful, and essentially surjective.

  Essential surjectivity is straightforward: given $\DeclVar{\prn{x,b}}{\GM}{\TotalTy{B}}$, we
  take $\prn{x,\alpha\prn{x}^{-1}\prn{b}}$ as $\alpha$ is invertible on $\GM$
  elements of $X$. To show that $\TotalTy{\alpha}$ is fully faithful, note that
  transport induces an equivalence between $\Int \to \TotalTy{A}$ and $\Sum{x : \Int \to X}
  \Sum{a_0 : A\prn{x\,0},a_1 : A\prn{x\,1}} \Hom[A\prn{x\,1}]{x_!a_0}{a_1}$. Since $\alpha$
  preserves cocartesian edges, it therefore suffices to show that $\Hom[A\prn{x_1}]{x_!a_0}{a_1}
  = \Hom[B\prn{x_1}]{\alpha\prn{x\,1,x_!a_0}}{\alpha\prn{x\,1,a_1}}$ for $\DeclVar{x}{\GM}{\Int \to X}$ and
  $\DeclVar{a_\epsilon}{\GM}{A\prn{x\,\epsilon}}$. This holds as
  $\alpha\prn{x\,1}$ is invertible.
\end{proof}

Next we show that every cocartesian family $\DeclVar{A}{\GM}{X \times \Int \to \Uni}$, where
$X$ is a category, is of the form $\Glue\prn{A\prn{-,0}, A\prn{-,1}, \lambda x.\,\prn{x,-}_!}$. To
this end, we note the following:
\begin{lemma}
  \label{lem:dua:transport-cocart}
  Given $A$ as above, the transport map
  $\alpha = \lambda x.\,\prn{x,-}_!$
  is a cocartesian functor $A\prn{-,0} \to A\prn{-,1}$.
\end{lemma}
\begin{proof}
  Unfolding, this follows from the 3-for-2 condition which holds for cocartesian
  arrows~\citep[Proposition 5.1.8]{buchholtz:2023}.
\end{proof}

Consequently, $B = \Glue\prn{A\prn{-,0},A\prn{-,1},\alpha}$ is a cocartesian family. Moreover, we
can produce a map of families $A \to B$:
\[
  \Con{glue}\prn{x,i,a} = \prn{x,i, \prn{\prn{x,i \lor -}_! a, \lambda p : i = 0.\,p_!a}}
\]
In words, we use cocartesian transport to move $a : A\prn{x,i}$ to $A\prn{x,1}$ and, if $i = 0$ to
begin with, record the original $a$ as well.
\begin{lemma}
  $\Con{glue} : \Prod{p : X \times \Int} A\prn{p} \to B\prn{p}$ is a cocartesian functor.
\end{lemma}
\begin{proof}
  Following \citet[Theorem 5.3.19]{buchholtz:2023}, to check that $\Con{glue}$ is cocartesian, it
  suffices to check that the Beck-Chevalley natural transformation is invertible. By
  \cref{lem:primer:pw-iso-to-iso}, it suffices to check this on $\GM$ elements where it is
  immediate.
\end{proof}

\begin{corollary}
  $\Con{glue}$ is an equivalence of cocartesian families.
\end{corollary}
\begin{proof}
  Applying \cref{lem:dua:fw-equiv}, it suffices to check this equivalence fiberwise on
  $\GM$-annotated elements $\DeclVar{\prn{x,i}}{\GM}{X \times \Int}$. In particular, it suffices to
  check that induces an equivalence on $\GM$-annotated elements of $\Int$ which, by
  \cref{ax:int-global-points}, consists only of $0$ and $1$. However, over $0$ and $1$ we see that
  $\Con{glue}$ is an equivalence: over $0$, this is immediate and over $1$ it follows from the
  observation that cocartesian transport over the identity arrow is the identity.
\end{proof}

Thus, a cocartesian family $\DeclVar{A}{\GM}{X \times \Int \to \Uni}$ is fully described by its
restrictions to $0$ and $1$ along with the associated transport map. Combining this classification
result in the case where $X = \ObjTerm{}$ with \cref{thm:cat:univ-cocart}, we obtain the
directed univalence principle:
\dua

We note that, in general, we actually obtain an equivalence between
$\Modify[\GM]{X \times \Int \to \Cat}$ and $\Modify[\GM]{\Sum{A_0,A_1 : \Cat^X} A_0 \CoCartTo A_1}$.
For what follows, we also require a similar accounting of cocartesian families
$\DeclVar{A}{\GM}{X \times \Delta^2 \to \Uni}$. The story plays out in much the same way; we define $B_1$
as the following iterated glue type:
\begin{align*}
  &B_0 = \Glue\prn{A\prn{-,\bar{0}}, A\prn{-,\bar{1}}, \lambda x.\,\prn{x,-}_!}
  \\
  &B_1 = \Glue\prn{B_0, A\prn{-,\bar{2}} \circ \Proj[0], \lambda x.\,\prn{x, \bar{1} \lor - \land \bar{2}}_! \circ \Proj[0]}
\end{align*}

Combining \cref{cor:cocartesian:proj-cocart} and \cref{lem:dua:transport-cocart} with
\cref{thm:cocartesian:glue-is-cocart}, we conclude that $B_1$ is a cocartesian family
$X \times \Int^2 \to \Uni$. We then take $B$ to be the restriction of $B_1$ to $X \times \Delta^2$.
The map of families $\Con{glue}$ considered previously easily generalizes to this $\Delta^2$ case
and we may prove the following:
\begin{lemma}
  The map $\Con{glue}_2 : \Prod{p : X \times \Delta^2} A\,p \to B\,p$ is cocartesian.
\end{lemma}
\begin{proof}
  Unfolding, this family sends $\prn{x, i, j}$ to the following type, where we write $\alpha$ for
  cocartesian transport $\bar{0} \to \bar{1}$ and $\beta$ for cocartesian transport $\bar{1} \to
  \bar{2}$:
  \begin{align*}
    \Sum{a_2 : A\prn{\bar{2},x}}
    j = 0 \to \Sum{a_{1} : \beta^{-1}\prn{a_2}}
    i = 0 \to \alpha^{-1}\prn{a_1}
  \end{align*}

  There is then a canonical assignment $\Con{glue}_2 : A\prn{c,i,j} \to B\prn{c,i,j}$ given as follows:
  \[
    \prn{x,i,j,a} \mapsto
    \prn{
      \prn{x,- \lor i, - \lor j}_!\prn{a},
      \lambda \_ : j = 0\,
      \prn{
        \prn{x,- \lor i, 0}_!\prn{a},
        \lambda \_ : i = 0.\,
        a
      }
    }
  \]
  We may then check directly that this assignment preserves cocartesian arrows by unfolding their
  construction in the $\Glue$ and checking this holds on global data once more. In the end, it
  amounts to the proof that $\Con{glue}$ preserves cocartesian edges; when restricted to the global edges
  $0 \le 1$, $1 \le 2$ and $2 \le 3$ we find that the above characterization of $B$ collapses
  to a single application of $\Glue$ along various cocartesian functors.
\end{proof}
Again, an application of \cref{lem:dua:fw-equiv} allows us to conclude that $\Con{glue}_2$ is an
equivalence. Combining these steps once more with \cref{thm:cat:univ-cocart} once more, we conclude
the following:
\begin{corollary}
  Cocartesian transport induces an equivalence

  \noindent
  $\Modify[\GM]{\Cat^{X \times \Delta^2}}
  \Equiv \Modify[\GM]{\Sum{A_0,A_1,A_2 : \Cat^X} A_0 \CoCartTo A_1 \times A_1 \CoCartTo A_2}$.
\end{corollary}

\subsection{\texorpdfstring{$\Cat$}{Cat} is Segal and Rezk}

Having characterized both cocartesian fibrations over $X \times \Delta^2$ and $X \times \Int$ for
all categories $X$, it is only slightly more work to prove that $\Cat$ is both Segal and Rezk.

\begin{lemma}
  $\Cat$ is Segal.
\end{lemma}
\begin{proof}
  We wish to show that $\Cat^{\Delta^2} \to \Cat^{\Lambda^2_1}$ is an equivalence. Using
  \cref{ax:cubes-separate}, it suffices to show that the following map is an equivalence for all
  $n$:
  $ 
    \Modify[\GM]{\Int^n \times \Delta^2 \to \Cat}
    \to 
    \Modify[\GM]{\Int^n \times \Lambda^2_1 \to \Cat}
  $.

  Note that $\Modify[\GM]{\Int^n \times \Lambda^2_1 \to \Cat}$ is equivalent to the
  following:
  \[
    \Modify[\GM]{\Int^n \times \Int \to \Cat} \times_{\Modify[\GM]{\Int^n \to \Cat}} \Modify[\GM]{\Int^n \times \Int \to \Cat}
  \]
  This follows from the fact that $\Modify[\GM]{-}$ has an internal right adjoint ($\Modify[\SM]{-}$)
  together with the fact that $\Lambda^2_1 = \Int \Pushout{\ObjTerm{}} \Int$. Next, since $\Int^n$
  is a category, the results of the previous section allow us to rephrase the above map into the
  following:
  \begin{align*}
    &\Modify[\GM]{\Sum{A,B,C: \Cat^{\Int^n}} A \CoCartTo B \times B \CoCartTo C} \to
    \\
    &
    \Modify[\GM]{\Sum{A,B : \Cat^{\Int^n}} A \CoCartTo B}
    \times_{\Modify[\GM]{\Cat^{\Int^n}}}
    \Modify[\GM]{\Sum{B,C : \Cat^{\Int^n}} B \CoCartTo C}
  \end{align*}
  This being an equivalence follows immediately from the fact that $\Modify[\GM]{-}$ preserves pullbacks by
  virtue of \cref{ax:crisp-ind}.
\end{proof}

Inspecting this proof, we see that if $P : \Lambda^2_1 \to \Cat$, the resulting composed edge
$f : \Int \to \Cat$ is a cocartesian fibration where $f\prn{0} = P\prn{\bar{0}}$,
$f\prn{1} = P\prn{\bar{2}}$ and cocartesian transport from $0$ to $1$ is the composite of
transporting in $P$ from $\bar{0}$ to $\bar{1}$ and then from $\bar{1}$ to $\bar{2}$.

\begin{lemma}
  $\Cat$ is Rezk.
\end{lemma}
\begin{proof}
  We wish to show that all synthetic isomophisms in $\Cat$ are equivalent to the identity. Once
  more, we apply \cref{ax:cubes-separate} showing that the restriction map
  $\Modify[\GM]{\Int^n \to \Cat} \to \Modify[\GM]{\Int^n \times \IntIso \to \Cat}$ is an
  equivalence. Commuting $\Modify[\GM]{-}$ with the pushout defining $\IntIso$ once more, we may instead consider
  the map:
  \begin{align*}
    &\Modify[\GM]{\Int^n \to \Cat}
    \\
    &\to \Modify[\GM]{\Cat^{\Int^n \times \Delta^2}}
    \times_{\dots}
    \Modify[\GM]{\Cat^{\Int^n \times \Int}}
    \times_{\dots}
    \Modify[\GM]{\Cat^{\Int^n \times \Delta^2}}
  \end{align*}
  Applying the results of the previous section and commuting $\Modify[\GM]{-}$ past pullbacks, we
  may recast the above:
  \begin{align*}
    &\Modify[\GM]{\Int^n \to \Cat}
    \\
    &\to 
    \Modify[\GM]{
      \Sum{A,B : \Cat^{\Int^n}}
      \Sum{f : A \CoCartTo B}
      \Sum{g,h : B \CoCartTo A}
      f \circ g = \ArrId{} \times h \circ f = \ArrId{}
    }
  \end{align*}
  However, since $\Cat$ is a subtype of $\Uni$ it satisfies the univalence axiom.
  The result then follows from the observation that the data imposed on $f$ ensures
  that it is a family of equivalences.
\end{proof}

\subsection{\texorpdfstring{$\Cat$}{Cat} is simplicial}

Our remaining task is to show that $\Cat$ is simplicial. Unfortunately, this is far from immediately
obvious. After all, $\Cat$ is defined using the amazing right adjoint to $\Int \to -$, which is the
primary source of non-simplicial types. Fortunately, in this case we may use
\cref{thm:cat:univ-cocart} to produce a left section to $\eta : \Cat \to \Simp \Cat$ and this
implies that $\eta$ is an equivalence~\citep[Lemma 1.20]{rijke:2020}.

To begin with, we record the following lemma:
\begin{lemma}
  The commuting square between $\Cat_\bullet = \prn{\Sum{A : \Cat} A} \to \Cat$
  and $\Simp \Cat_\bullet \to \Simp \Cat$ induced by $\eta$ is cartesian.
\end{lemma}
\begin{proof}
  Unfolding definitions and comparing fibers, this follows from the fact that each
  $A : \Cat$ is simplicial and $\Simp$ is lex.
\end{proof}

Next, we note that if $\Simp \Cat_\bullet \to \Simp \Cat$ represents a cocartesian family itself,
then there must be a unique classifying map $\Simp \Cat \to \Cat$ and, pasting together the relevant
pullback squares, we obtain the following composite pullback diagram:
\[
  \begin{tikzpicture}[diagram]
    \SpliceDiagramSquare{
      nw/style = pullback,
      ne/style = pullback,
      height = 1.2cm,
      nw = \Cat_\bullet,
      sw = \Cat,
      ne = \Simp \Cat_\bullet,
      se = \Simp \Cat,
    }
    \node [right = of ne] (A) {$\Cat_\bullet$};
    \node [right = of se] (B) {$\Cat$};
    \path [->] (ne) edge (A);
    \path [->] (se) edge (B);
    \path [->] (A) edge (B);
  \end{tikzpicture}
\]
By the univalence property of $\Cat$, this bottom composite must be the identity. Consequently, the
classifying map $\Simp \Cat \to \Cat$ is the required left inverse to the unit. All that remains,
therefore, is to prove that $\Simp \Cat_\bullet \to \Simp \Cat$ is a cocartesian fibration.

To prove this we will use \cref{ax:simp-stability} along with
\cref{lem:primer:pw-left-adjoint-to-left-adjoint,prop:cocartesian:lari}.
First, we must note that $\Simp \Cat$ and $\Simp \Cat_\bullet$ are both
categories by \cref{cor:proofs:simp-cat}.

\begin{lemma}
  The family $\Simp \Cat_\bullet \to \Simp \Cat$ is cocartesian.
\end{lemma}
\begin{proof}
  As a map between categories, it is automatically iso-inner and simplicial. It therefore suffices
  to prove that the comparison map $\prn{\Simp \Cat_\bullet}^\Int \to \prn{\Simp \Cat}^\Int
  \times_{\Simp \Cat} \Simp \Cat_\bullet$ has a left adjoint right inverse. For concision, we write
  $C = \Simp \Cat$ and, commute $\Simp$ with $\Sum{}$ to replace $\Simp \Cat_\bullet$ with
  $E = \Sum{A : \Simp \Cat} \DepSimp A$ in what follows.

  We may use \cref{lem:primer:pw-left-adjoint-to-left-adjoint} so that it suffices to check that
  this property holds when restricted to $\GM$-annotated elements of  $C^\Int \times_{C} E$.
  Accordingly, we may fix $\DeclVar{A}{\GM}{\Int \to C}$ along with an element
  $\DeclVar{a_0}{\GM}{A\,0}$. By \cref{ax:simp-stability}, we may assume that $A = \eta \circ A'$
  and $a_0 = \eta\prn{a_0'}$ for some unique $\DeclVar{A'}{\GM}{\Int \to \Cat}$ and
  $\DeclVar{a_0'}{\GM}{A}$

  Since $\Cat_\bullet \to \Cat$ is cocartesian, we lift $A',a_0'$ to a cocartesian
  morphism $\DeclVar{a'}{\GM}{\prn{i : \Int} \to A'\,i}$ with $p' : a'\,0 = a_0'$. We now
  argue that $a = \eta \circ a' : \prn{i : \Int} \to \DepSimp A\,i$ 
  and the induced equality $p : a\,0 = a_0$ is the desired universal lift over $A$, $a$.

  After appropriately massaging the input data, it suffices to show that if we are given
  $\DeclVar{H}{\GM}{\Delta^2 \to \Simp \Cat}$ and a lift of this map $\DeclVar{h}{\GM}{\prn{t :
  \Lambda^2_0} \to \DepSimp H\prn{t}}$. There are a handful of paths relating these two $a$ and $A$,
  but after induction we may definitionally identify (1) $H\prn{0,-}$ with $A$ (2) $h\prn{0,-}$ with
  $a$. We must show that $h$ extends uniquely to some $\bar{h} : \prn{t : \Delta^2} \to \DepSimp
  H\prn{t}$. By \cref{ax:simp-stability}, we may once more factor $H$ and $h$ through
  $\Cat$ whereby the result is an immediate consequence of our construction of $a'$ as a
  cocartesian lift.
\end{proof}

\begin{corollary}
  $\Cat$ is a category.
\end{corollary}

\section{Full straightening-unstraightening}
\label{sec:st-unst}

In this section, we prove the Lurie's straightening--unstraightening theorem which states that for a
category $\DeclVar{C}{\GM}{\Uni}$, the type $C \to \Cat$ is equivalent to the subcategory of
$\SLICE{\Cat}{C}$ (that is, $\Sum{f : \Cat^\Int} f\prn{1} = C$) restricted such that its 0- and
1-cells are given cocartesian families and cocartesian functors. Accordingly, for the remainder of
this section let us fix $\DeclVar{C}{\GM}{\Uni}$ a category.

To do this, we will construct a map $U : \prn{C \to \Cat} \to \SLICE{\Cat}{C}$ and prove that it is (1)
an embedding such that (2) its image on $\GM$-annotated elements $C \to \Cat$, and $\Int \to \prn{C
\to \Cat}$ satisfies precisely the above criteria. From this, we show that $\prn{C \to \Cat} \to
\SLICE{\Cat}{C}$ satisfies the expected universal property for the subcategory $\CoCart{C}$ of
\emph{cocartesian families over $C$} (\cref{cor:st-unst:st-unst}).

\begin{remark}
  The material in this section closely follows \citet{cisinski:2024} with only minor alterations to
  make it more convenient in \TTT{}. In particular, it is from there that we learned of this
  method of constructing the unstraightening functor and characterizing its image. That such an
  adaptation is possible is expected but encouraging: the axiomatic approach given by \opcit{} is
  intended to give high-level arguments which can be translated into formal systems satisfying their
  axioms and our construction of $\Cat$ ensures that \TTT{} satisfies all the relevant axioms for
  this argument.
\end{remark}

\begin{remark}
  We avoid explicitly constructing $\CoCart{C}$ merely to avoid the
  detour of describing the construction of non-full subcategories. Such constructions are possible
  using \eg{}, $\prn{-}_\Int$.
\end{remark}

\subsection{The unstraightening map}

We begin by constructing a map $U$ from $C \to \Cat$ to $\SLICE{\Cat}{C}$. We break this process
into several steps. We begin by considering two particular cocartesian families over $C \to \Cat$:
\[
  E\prn{f} = \Sum{c : C} f\prn{c}
  \qquad
  B\prn{f} = C
\]
These are both cocartesian families over $C$ and the canonical projection is a cocartesian map
$\Proj[0] : E \CoCartTo B$ as well---and therefore a cocartesian functor. We may
therefore glue these together to form a cocartesian family:
$\Glue\prn{E,B,\Proj[0]} : \prn{\Cat^C} \times \Int \to \Cat$. First, let us compute
$\Glue\prn{E,B,\Proj[0]}\prn{-,1}$ and observe that it is canonically identified to $\lambda
\_.\,C$, via the following family of paths:
\[
  \Phi = \lambda f.\,\Con{ua}\prn{\Proj[1]} : \Prod{f : C \to \Cat} \Glue\prn{E,B,\Proj[0]}\prn{f,1} = C
\]
By transposing and using this identification, we obtain:
\begin{align*}
  &U : \prn{C \to \Cat} \to \SLICE{\Cat}{C}
  \\
  &U = \lambda f.\,\prn{\Glue\prn{E,B,\Proj[0]}\prn{f,-}, \Phi\prn{f}}
\end{align*}

\subsection{The image of the unstraightening map}

Our next task to identify the image of $U$ and, in particular, to show that it is precisely the
category of cocartesian families over $C$ and cocartesian functors between them. We therefore
compute fibers of $\Modify[\GM]{\Cat^C} \to \Modify[\GM]{\SLICE{\Cat}{C}}$
and $\Modify[\GM]{\Int \to \Cat^C} \to \Modify[\GM]{\Int \to \SLICE{\Cat}{C}}$.

In particular, we will show that the fiber over $\DeclVar{p}{\GM}{\SLICE{\Cat}{C}}$ is precisely the
proposition stating whether $p$ is cocartesian and over $\DeclVar{f}{\GM}{\Int \to \SLICE{\Cat}{C}}$
it corresponds to the triple of propositions requiring $f\prn{0}$ and $f\prn{1}$ to be cocartesian
and $f$ itself to induce a map of cocartesian families. In light of the Segal condition, this
characters the fibers over arbitrary simplices and, via \cref{ax:cubes-separate} and the
simpliciality of $\Cat$, proves that $U$ is an embedding. Moreover, our description of the fibers
shows that the resultant subcategory of $\SLICE{\Cat}{C}$ is precisely as described at the beginning
of this section.

\begin{lemma}
  The fiber of $U$ over $\DeclVar{p}{\GM}{\SLICE{\Cat}{C}}$ is a proposition inhabited iff $p$ is cocartesian.
\end{lemma}
\begin{proof}
  Post-composing with directed univalence, we may identify a fiber of $U$ with a fiber of
  the map:
  \[
    U' : \Modify[\GM]{C \to \Cat} \to \Modify[\GM]{\Sum{E : \Cat} E \to C}
  \]
  Consider a category $\DeclVar{E}{\GM}{\Cat}$ and a function $\DeclVar{\pi_E}{\GM}{E \to C}$. Our goal
  is to compute the fiber $\Sum{f : \Modify[{\GM}]{C \to \Cat}} U'\prn{f} = \MkMod[\GM]{E,\pi_E}$.

  Unfolding $U'$, an element $\DeclVar{f}{\GM}{C \to \Cat}$ is sent to $\prn{E,\pi_E}$ if and only if we have
  an equivalence $\DeclVar{e}{\GM}{E \Equiv \Sum{c : C} f\prn{c}}$ and an equation 
  $\DeclVar{p}{\GM}{\pi_E = \Proj[1] \circ e}$. By another application of univalence, this is
  equivalent to requiring that $\Modify[\GM]{\pi_E^{-1}\prn{-} = f}$, so the fiber amounts to
  $\Sum{\DeclVar{f}{\GM}{C \to \Cat}} \Modify[\GM]{\pi_E^{-1} = f}$. This is a proposition by
  \cref{prop:primer:modalities} and by \cref{thm:cat:univ-cocart} inhabited iff $\pi_E$ is a
  cocartesian family.
\end{proof}

\begin{lemma}
  The fiber of $U$ over $\DeclVar{f}{\GM}{\Int \to \SLICE{\Cat}{C}}$ is a proposition inhabited iff
  $f$ is is a cocartesian functor between cocartesian families.
\end{lemma}
\begin{proof}
  Rearranging equations, we may identify $\Modify[\GM]{\Int \to \SLICE{\Cat}{C}}$ with
  $\Modify[\GM]{\Sum{h : \Delta^2 \to \Cat} h\prn{\bar{2}} = C}$ which we may then identify via
  directed univalence with
  $\Modify[\GM]{\Sum{E,F : \Cat} E \to F \times F \to C}$. Post-composing $U$ with these
  maps, we instead consider the following:
  \[
    U' :
    \Modify[\GM]{\Int \to \Cat^C} \to
    \Modify[\GM]{\Sum{E,F : \Cat} E \to F \times F \to C}
  \]
  Consider categories $\DeclVar{E,F}{\GM}{\Cat}$ and functions $\DeclVar{\pi_F}{\GM}{F \to C}$ and
  $\DeclVar{\alpha}{\GM}{E \to F}$. Our goal is to compute the fiber of $U'$ over this data.

  Unfolding, an element $\DeclVar{h}{\GM}{\Int \to \Cat^C}$ is sent to $E,F,\pi_F$ if and only if we
  have the following:
  \begin{itemize}
    \item an equivalence $\DeclVar{e_0}{\GM}{E \Equiv \Sum{c : C} h\prn{0,c}}$,
    \item an equivalence $\DeclVar{e_1}{\GM}{F \Equiv \Sum{c : C} h\prn{1,c}}$,
    \item a path $\DeclVar{\phi}{\GM}{e_1 \circ \alpha = \beta \circ e_0}$ where
      $\beta = \lambda c.\,\prn{-,c}_!$ is given by the cocartesian transport of $h$
    \item a path $\DeclVar{\phi}{\GM}{\pi_F = \Proj[0] \circ e_1}$.
  \end{itemize}

  Once more using univalence, we are reduced to $\GM$-annotated equation between maps of families
  over $C$ (\ie{} $C \to \Sum{A, B : \Uni} B^A$):
  \begin{align*}
    &\lambda c.\,\Proj[0] :
    \Prod{c : C} \prn{\Sum{f : \pi_F^{-1}\prn{c}} \alpha^{-1}\prn{f}} \to \pi_F^{-1}\prn{c}
    \\
    &\lambda c.\,\prn{-,c}_! : \Prod{c : C} h\prn{0,c} \to h\prn{1,c}
  \end{align*}
  Since $\Modify[\GM]{\Int \to \Cat^C}$ embeds into $\Modify[\GM]{C \to \Sum{A, B : \Uni} B^A}$ by
  directed univalence, \cref{prop:primer:modalities} ensures this fiber over $\prn{E,F,\pi_F}$ is a
  proposition. The same analysis shows that it is inhabited iff $\pi_F \circ \alpha$ and
  $\pi_E$ are cocartesian and $\alpha$ is a map of cocartesian families.
\end{proof}

\begin{corollary}
  The map $U : \prn{C \to \Cat} \to \SLICE{\Cat}{C}$ is an embedding.
\end{corollary}
\begin{proof}
  To show that $U$ is an embedding, we will show that
  $\Delta_U : \prn{C \to \Cat} \to \prn{C \to \Cat} \times_{\SLICE{\Cat}{C}} \prn{C \to \Cat}$
  is an equivalence. Applying \cref{ax:cubes-separate} along with the fact that all of the objects
  involved here are simplicial, it suffices to show that the following map is an equivalence:
  \[
    \Modify[\GM]{\Delta^n \to \Cat} \to
    \Modify[\GM]{\Delta^n \to \prn{C \to \Cat} \times_{\SLICE{\Cat}{C}} \prn{C \to \Cat}}
  \]

  Since both sides of this are categories, we may restrict to the case where $n = 0,1$. In this
  case, it suffices to show that if
  $\DeclVar{f,g}{\GM}{\Delta^n \to \prn{C \to \Cat}}$ and $\DeclVar{p}{\GM}{U \circ f = U \circ g}$
  then there is a path $(f,\Refl) = (g, p)$. However, this is precisely equivalent to asking that
  the fiber of $\prn{U \circ -}^\dagger$ over $\MkMod[\GM]{U \circ f}$ is contractible. By previous
  results, we know the fiber is a proposition and it is inhabited by $\prn{f, \Refl}$. Consequently,
  it is contractible as required.
\end{proof}
In fact, in light of our identification of the fibers of $U$ we may also characterize when a functor
lifts along it. This is, in essence, the universal property of $\CoCart{C}$ mentioned earlier:
\begin{corollary}[Straightening--unstraightening]
  \label{cor:st-unst:st-unst}
  If $\DeclVar{D}{\GM}{\Uni}$ is a category, a map $\DeclVar{f}{\GM}{D \to \SLICE{\Cat}{C}}$ lifts
  along $U$ to $\Cat^C$ if and only if
  \begin{enumerate}
    \item for each $\DeclVar{d}{\GM}{D}$, the functor $f\prn{d}$ is a cocartesian family.
    \item for each $\DeclVar{d}{\GM}{\Int \to D}$, the functor induced by $f \circ d : \Int \to
      \SLICE{\Cat}{C}$ is a cocartesian functor between the cocartesian families.
  \end{enumerate}
\end{corollary}
\begin{proof}
  Our goal is to characterize for which $f$ the following map is an equivalence:
  \[
    D \times_{\SLICE{\Cat}{C}} \prn{C \to \Cat} \to D 
  \]
  Notably, we know already this map is an embedding (it is the pullback of $U$) and so we merely
  wish to characterize when it is surjective. Using \cref{ax:cubes-separate} along with the fact
  that both sides are categories, it suffices to consider when the following maps are surjective:
  \begin{align*}
    &\Modify[\GM]{D \times_{\SLICE{\Cat}{C}} \prn{C \to \Cat}} \to \Modify[\GM]{D}
    \\
    &\Modify[\GM]{\Int \to D \times_{\SLICE{\Cat}{C}} \prn{C \to \Cat}} \to \Modify[\GM]{\Int \to D}
  \end{align*}

  We now unfold these maps and use \cref{prop:primer:modalities}. These guarantee that the first map
  will hit $\DeclVar{d}{\GM}{D}$ if and only if $f\prn{d}$ is a cocartesian family. Similarly,
  the second map will hit $\DeclVar{d}{\GM}{\Int \to D}$ if and only if $f \circ d$ is a cocartesian
  functor between cocartesian families.
\end{proof}

\section{Examples}
\label{sec:examples}

We have thus far focused on the construction of $\Cat$ and verifying its essential
properties and so we close by discussing some of the new examples and category theory unlocked by $\Cat$.
We content ourselves with only sketching several examples.

\subsection{Subcategories of \texorpdfstring{$\Cat$}{Cat}}
We begin by noting that since every covariant family is cocartesian, there is a unique map from the
base of the universal covariant family $\Space$ to the base of the universal cocartesian family
$\Cat$. This is the inclusion of groupoids into categories.
\begin{lemma}
  \label{lem:examples:space-cohesion}
  The map $i : \Space \to \Cat$ is fully faithful and possesses both left and right adjoints:
  $\vrt{-} \Adjoint i \Adjoint \prn{-}^\Equiv$.
\end{lemma}
\begin{proof}[Proof Sketch]
  The second half of this statements follows from
  \cref{lem:primer:pw-left-adjoint-to-left-adjoint}. In particular, we use this lemma to extend the
  pointwise assignments of $\DeclVar{X}{\GM}{\Cat} \mapsto \Modify[\GM]{X} : \Space$ and
  $\DeclVar{X}{\GM}{\Cat} \mapsto \Grpdify X : \Space$ to functors $\Cat \to \Space$. The fact that $i$
  is fully faithful is then immediate from \cref{ax:int-detects-discrete}: if $X$ is a groupoid then
  the unit $X \to \Modify[\GM]{i\prn{X}}$ is an equivalence. It is a standard argument that
  a unit being invertible implies the left adjoint is fully faithful.
\end{proof}

Many other interesting categories exist as full subcategories of $\Cat$. For instance, we may
isolate univalent 1-categories as the full subcategory of $\Cat$~\citep[\S 7]{gratzer:2024} given by the
following predicate:
\begin{align*}
  &\Con{is1Cat} : \MFn[\GM]{\Cat}{\Prop}
  \\
  &\Con{is1Cat}\prn{C} = \Prod{\DeclVar{a,b}{\GM}{C}} \IsSet\prn{\Hom[C]{a}{b}}
\end{align*}

Similar definitions immediately yield $(n,1)$-categories for all $n$. Notably, by restricting to $n
= -1$ we obtain the category of partial orders and, restricting further to linear partial orders,
the simplex category $\SIMP \hookrightarrow \Cat$. In fact, the same argument as was used to $\Space
\hookrightarrow \Cat$ allows us to prove the following:
\begin{lemma}\label{lem:examples:n-cat-to-cat}
  The inclusion of $\Cat_n \hookrightarrow \Cat$ is a right adjoint. 
\end{lemma}
\begin{proof}[Proof Sketch]
  One adapts \cref{lem:examples:space-cohesion} to use the modality nullifying the maps $\Lambda^2_1
  \to \Delta^2$, $\IntIso \to \ObjTerm{}$, and $\partial \Delta^{n + 2} \to \Delta^{n + 2}$.
\end{proof}

\NewDocumentCommand{\FinSet}{}{\Con{Fin}}

\paragraph{Towards algebraic K-theory}
For a small example of how these ingredients might be combined to build a useful and important
construction in higher category theory, we turn our attention to monoidal categories. 
Let us write $\brk{n}$ for the element of $\SIMP$ realizing the linear order
$\brc{0 \le \dots \le n}$. Using \cref{cor:dua:dua}, we define
$\rho_n^i : \Hom{\brk{1}}{\brk{n}}$ which sends $0 \le 1$ to
$\brc{i \le i + 1} \subseteq \brk{n}$.
\begin{definition}
  A monoidal category $C^\otimes : \Cat^{\Modify[\OM]{\SIMP}}$ is a functor
  where
  $\prn{\rho^1_n, \dots, \rho^n_n} : C^\otimes\prn{\brk{n}} \to C^\otimes\prn{\brk{1}} \times \dots \times C^\otimes\prn{\brk{1}}$
  is an equivalence for all $n$.
\end{definition}

\noindent
Replacing $\Cat$ by $\Space$ in the above gives the definition of an  $E_1$-monoid: a
homotopy-coherent monoid~\citep[\S 7]{gratzer:2024}.

\begin{definition}
  The category of monoidal categories $\Con{MCat}$ is the full subcategory of
  $\Modify[\OM]{\SIMP} \to \Cat$ spanned by monoidal categories.
\end{definition}

We readily adapt this definition to (1) the category of $E_1$-monoids $\Con{Mon}$ as a subcategory
of $\Space^{\Modify[\OM]{\SIMP}}$ and to (2) the category of monoidal 1-categories $\Con{MCat}_1$ as
a subcategory of $\Cat_1^{\Modify[\OM]{\SIMP}}$.

As both $\prn{-}^\Equiv : \Cat \to \Space$ and the inclusion $\Cat_1 \to \Cat$ are right adjoints,
they preserve finite products and therefore post-composing by these maps induces functors
$\Con{MCat} \to \Con{Mon}$ and $\Con{MCat}_1 \to \Con{MCat}$. We note next that---viewing $\Con{Mon}$
as a subcategory of $\Space^{\Modify[\OM]{\SIMP}}$---we may take the colimit of 
$M : \Con{Mon}$ to obtain a space $\Colim{} M$. In fact, this space is canonically pointed: the
initial object in $\Con{Mon}$ is the functor $\Const\,\ObjTerm{}$ and
$\Colim{} \Const\,\ObjTerm{} = \ObjTerm{}$. Finally, regarding the loop-space functor as a map
$\Omega : \Space_* \to \Space_*$ we define $k$ to be the following chain of functors:
\[
  k :
  \Con{MCat}_1 \to
  \Con{MCat} \to
  \Con{Mon} \to
  \Space_* \to
  \Space_*
\]

We may now define the simplest form of algebraic K-theory:
\begin{definition}[Quillen]
  The \emph{$i$th K-group} of a monoidal 1-category $C^\otimes$ is the $i$th homotopy group
  $K_i\prn{C^\otimes} = \pi_i\prn{k\prn{C^\otimes}}$.
\end{definition}

Notably, with $\Cat$ to hand all of these definitions are quite conceptual and automatically
functorial. We emphasize that this is only a first step towards realizing K-theory. We leave it to
future work to show \eg{}, that modules over a ring are a monoidal category.

\subsection{The structure homomorphism principle}

To give a different class of examples involving $\Cat$, we turn to the \emph{structure homomorphism
principle}. This is the directed enhancement of \HOTT{}'s structure identity principle. This
principle states that by taking ordinary type-theoretic definitions of objects in a certain
category but using $\Cat$ or $\Space$ instead of $\Uni$, we obtain the correct synthetic category
with the expected homomorphisms.

For a simplest example of this phenomena:
\begin{lemma}
  The type $\Sum{A : \Cat} A$ is the \emph{lax} slice $\ObjTerm{} \dblslash \Cat$ \ie{} its objects
  are pointed categories $\prn{C,c}$ and when $\DeclVar{\prn{C,c},\prn{D,d}}{\GM}{\Sum{A : \Cat} A}$
  morphisms $\Hom{\prn{C,c}}{\prn{D,d}}$ consist of functions $f : C \to D$ together with a morphism
  $\Hom{f\prn{c}}{d}$.
\end{lemma}
\begin{proof}
  As $\Sum{A : \Cat} A$ is the total space of the cocartesian family $\Cat \to \Uni$, it is a
  category. The characterization of objects is immediate from \cref{prop:primer:modalities}. For
  morphisms, we use \cref{cor:dua:dua} with the factorization of a morphism in a total space of a
  cocartesian family into a cocartesian map followed by a vertical map.
\end{proof}

More generally, one may show that $\Sum{A : \Cat} A^C$ is the lax slice category $C \dblslash \Cat$
for any $\DeclVar{C}{\GM}{\Cat}$.

For a more sophisticated example, consider the following type built using the right adjoint to $i :
\Space \to \Cat$.
\[
  \Cat_{\mathrm{smarked}} = \Sum{C : \Cat} \prn{\prn{\Int \to C}^\Equiv \to \Bool}
\]
By directed univalence applied to $\Space$, this type is equivalent to $\Sum{C : \Cat}
\Hom{\prn{\Int \to C}^\Equiv}{\Bool}$ which is easily seen to be a category.
Objects of $\Cat_{\mathrm{smarked}}$ are a pairs of (1) a category $\DeclVar{C}{\GM}{\Cat}$ and (2)
a (decidable) predicate $\phi_C$ on the groupoid $\Modify[\GM]{\Int \to C}$ (recall that
$\prn{-}^\Equiv$ extends $\Modify[\GM]{-}$). The morphisms in $\Cat_{\mathrm{smarked}}$ are then
functors between the underlying category $f : C \to D$ such that $\phi_C = \phi_D \circ f$. This is
almost the category of \emph{marked} categories (the category of categories with a distinguished class of
morphisms along with functors preserving these morphisms) but the morphisms are off: we should
expect only that $\phi_C$ implies $\phi_D \circ f$.

To rectify this, we should replace $\Bool$ with the non-discrete category $\Int$. Just as with
$\Sum{A : \Cat} A$, this introduces the required laxity in the morphism. Unfortunately, however, we
can no longer rely on directed univalence for $\Space$ in this case; we must use \cref{cor:dua:dua}
which only applies to $\GM$-annotated elements. This makes it much more difficult to prove
that $\Sum{C : \Cat} \prn{\prn{\Int \to C}^\Equiv \to \Int}$ is a category. Systematically handling
these ``mixed-variance'' applications of SHP requires more exploration of exponentiable
functors~\citep{bardomiano:2025:exponentiable}: we must show that (co)cartesian functors are
exponentiable. We defer this to future work and instead work relative to the following conjecture:

\begin{conjecture}\label{conj:examples:exponentiable} If $\DeclVar{C}{\GM}{\Cat}$ then $- \to C :
  \Space \to \Cat$ is cartesian where the cartesian lift of $\DeclVar{f}{\GM}{\Hom[\Space]{A}{B}}$
  to $\prn{B,\phi}$ is $\prn{A, \phi \circ f}$.
\end{conjecture}
The difficult piece of this conjecture is establishing that $- \to C$ is an iso-inner family, with
the cartesian lifts following directly from the dual of \cref{prop:cocartesian:lari}. To see this
enables further applications of SHP, note that choosing $C = \Int$ yields:
\begin{lemma}
  The category $\Cat_{\mathrm{marked}} = \Sum{C : \Cat} \prn{\prn{\Int \to C}^\Equiv \to \Int}$ is
  the category of marked categories.
\end{lemma}
\begin{proof}
  We discuss only the characterization of synthetic morphisms. First, we note that the projection
  $\Proj[0] : \prn{\Sum{C : \Cat} \prn{\prn{\Int \to C}^\Equiv \to \Int}} \to \Cat$ is a pullback of the family described in
  \cref{conj:examples:exponentiable} and therefore cartesian. We may then factor any morphism
  $\DeclVar{f}{\GM}{\Hom[\Cat_{\mathrm{marked}}]{\prn{C,\phi_C}}{\prn{D,\phi_D}}}$ uniquely into a
  cartesian lift of a map $f_0 : C \to D$ along with a morphism $f_1 : \Hom[\Int]{\phi_C}{\phi_D
  \circ \prn{\prn{f_0}_*}^\Equiv}$. This is precisely the data of a morphism of marked categories.
\end{proof}

\section{Conclusions and related work}
\label{sec:conc}

In this work, we have constructed a subtype of the universe $\Cat \hookrightarrow \Uni$ and shown
that it gives rise to the \emph{category of categories} within \STT{}. In particular, we have shown
that it is simplicial, Segal, and Rezk and characterized its objects and mapping spaces to be
categories and functors. To give an even more precise characterization of $\Cat$, we then prove
Lurie's straightening--unstraightening theorem. Finally, we show how results from
$\infty$-category theory can now be proved (\eg{}, various adjunctions between $\Space$ and
$\Cat$) in short order and use the particular nature of type theory (in the form of the
SHP) to construct new $\infty$-categories quickly and intuitively.

\subsection{Related work}

There has been a large amount of work on both directed approaches to type theory and
straightening--unstraightening generally. Our work fits into the broader tradition of directed type
theory, specifically the line of work initiated by \citet{riehl:2017} on simplicial type theory~%
\citep{%
  bardomiano:2025:limits,%
  bardomiano:2025:exponentiable,%
  buchholtz:2023,%
  riehl:2025,%
  shulman:2019,%
  weinberger:twosided:2024,%
  weinberger:sums:2024,%
  weinberger:phd,%
  gratzer:2025%
}.

In this specific context, $\Cat$ fills a major missing piece in the foundations of \STT{} and allows
for new applications to, for instance, higher algebra and monoidal category theory. For directed
type theory generally, we believe this construction should make it possible to more directly provide
semantics to ``fully directed'' type theories~%
\citep{%
  licata:2011,%
  warren:2013,%
  nuyts:2015,%
  north:2018,%
  kavvos:directed:2019,%
  nuyts:2020,%
  ahrens:2023,%
  neumann:2024,%
  neumann:2025,%
  neumann:phd%
  }.
In particular, a model of those type theories which adequately models $\infty$-categories can now be
constructed by a syntactic translation to \TTT{}. The presence of $\Cat$, in particular, means that
the features supported by \TTT{} should be adequate to encode all of those presently considered.

\paragraph{Other approaches to amazing right adjoints}
The use of the right adjoint to $\Int \to -$ to construct a universal family with suitable structure
dates back to \citet{licata:2018}, where it was used in the context of cubical type theory. While,
for instance, \citet{riley:2024} proposed a more refined version of this technique which gave a more
judgmental account of the amazing right adjoint, our approach closely follows \citet{licata:2018}
and, especially, its adaptation to directed type theory by \citet{weaver:2020}. In particular,
\opcit{} use this methodology to construct a directed univalent universe of groupoids in
bicubical type theory. \citet{gratzer:2024} then introduced \TTT{} in order to adapt this
methodology to apply to standard $\infty$-categories and further prove that the universe of
groupoids is a category. Our work carries this program further forward by constructing not just a
category of groupoids via a universal covariant family, but a category of categories via a universal
cocartesian family. This is a substantial step---cocartesian families enjoy fewer nice properties
than covariant families and this complicates various aspects of the proof \eg{}, the argument
that $\Cat$ is simplicial.

\paragraph{Existing proofs of straightening--unstraightening}
Since Lurie's original proof~\citep{lurie:2009}, two additional proofs have been published in the
literature. The first is given by \citet{hebestreit:2025} and is a much more directed and succinct
version of Lurie's proof, but still employs the same fundamental approach. In
particular, the main result is a Quillen equivalence of two model categories rather than an
intrinsically $\infty$-categorical or synthetic approach.

More closely related is the work of \citet{cisinski:2022} which develops a universal cocartesian
family as we do and uses this to prove straightening--unstraightening, directed univalence, and
similar. The proof strategy of our approach is quite similar to \opcit{}: a straightforward proof of
the universal fibration and a more lengthy argument that its base is a category. In fact, \opcit{}
is influenced by various standard techniques in the semantics of type theory. A crucial difference
between our approaches, however, is the ambient framework used to manage $\infty$-categories. Unlike
our approach, \citet{cisinski:2022} use quasicategories along with various model categorical tools
from the Joyal and marked model structures. This is in contrast with our more synthetic approach,
which uses model categories only indirectly via \cref{thm:primer:soundness}.

Particularly in light of the last proof, we feel that our approach is an interesting example in
how higher category theory can influence the development of type theory (through \HOTT{}) which can
in turn contribute new techniques and proofs to higher category theory. We also hope that the
model-independent nature of our proof will result in a result which can be applied to exotic
situations such as internal $\infty$-categories in an $\infty$-topos~%
\citep{%
  martini:2022a,%
  martini:2025%
}.

\subsection{Future work}

While this work has constructed the category of categories, much work still remains to be done
around (1) the use of $\Cat$ in synthetic $\infty$-category theory
and (2) deriving a computational account of this type theory, following on the work of
\citet{weaver:2020}. For the first point, we believe that higher algebra and the theory of operads
is particularly interesting to study in this regard, as the existing foundations of the
theory~\citep{lurie:2017} are notoriously technical and rely on intricate simplex-by-simplex
arguments. For the second, we believe that the machinery of our argument could be adapted to
bicubical type theory which would, at the very least, provide a constructive model of our argument.

\bibliographystyle{ACM-Reference-Format}
\bibliography{../refs,../temp-refs}

\newpage
\appendix
\section{Formal syntax of \texorpdfstring{\MTT{}}{MTT}}
\label{sec:mtt-syntax}

We provide a succinct description of the formal syntax of \MTT{} in this section. Since most
connectives of type theory ($\Sum{}$, $\Coprod{}$, \etc{}) are not impacted by modalities, we focus
only on those rules which must be changed. These are (1) some aspects of the substitution calculus
and (2) the rules for modal types are modal $\Prod{}$ types. We assume a mode theory $\Mode$ which
has 1 object and which is enriched in posets, as is in this paper.

First, we extend contexts with the following new forms:
\begin{mathpar}
  \inferrule{
    \IsCx{\Gamma}
    \\
    \mu : m \to m \in \Mode
  }{
    \IsCx{\LockCx{\Gamma}}
  }
  \and
  \inferrule{
    \IsCx{\Gamma}
    \\
    \IsTy[\LockCx{\Gamma}]{A}
  }{
    \IsCx{\MECx{\Gamma}{A}}
  }
\end{mathpar}

Our previous notation with formal divisions was really syntactic sugar for these operations. In
particular, $\DeclVar{x}{\mu/\nu}{A}, \DeclVar{y}{\ArrId{}/\nu}{B}, z : C$ becomes
$\LockCx{\MECx{\MECx{\EmpCx}[\mu]{A}}[\ArrId{}]{B}}[\nu].C$. Notably, while $-/\mu$ was mere
notation for the paper, it is actually the primitive operation in \MTT{}. Any context built using
either notation using the rules of the system is translatable.

We then add several new to the substitution calculus to account for this. This includes a new form
of the variable rule (built using de Bruijn indices) to account for $\LockCx{\Gamma}$.

\begin{mathpar}
  \inferrule{
    \IsSb[\Delta]{\gamma}{\Gamma}
  }{
    \IsSb[\LockCx{\Delta}]{\LockSb{\gamma}}{\LockSb{\Gamma}}
  }
  \and
  \inferrule{
    \IsCx{\Gamma}
    \\
    \mu \le \nu
  }{
    \IsSb[\LockCx{\Gamma}[\nu]]{\KeySb{\mu \le \nu}{\Gamma}}{\LockCx{\Gamma}[\mu]}
  }
  \and
  \inferrule{
    \IsTy[\LockCx{\Gamma}]{A}
  }{
    \IsSb[\MECx{\Gamma}{A}]{\Wk}{\Gamma}
  }
  \and
  \inferrule{
    \IsTy[\LockCx{\Gamma}]{A}
    \\
    \IsSb[\Delta]{\gamma}{\Gamma}
    \\
    \IsTm[\Delta]{M}{\Sb{A}{\LockSb{\gamma}}}
  }{
    \IsSb[\Delta]{\ESb{\gamma}{M}}{\MECx{\Gamma}{A}}
  }
  \and
  \inferrule{
    \IsTy[\LockCx{\Gamma}]{A}
  }{
    \IsTm[\LockCx{\MECx{\Gamma}{A}}]{\mathbf{var}}{\Sb{A}{\LockSb{\Wk}}}
  }
\end{mathpar}

We have normal substitution rules around substitution extensions and weakenings. These are
essentially standard, and so we omit them. We further require a handful of equations which ensure
that $\Gamma \mapsto \LockCx{\Gamma}$, $\gamma \mapsto \LockSb{\gamma}$, and $\KeySb{- \le -}{-}$
organize into a 2-functor from $\Coop{\Mode}$ to $\CAT$ sending $m$ to the category of contexts. We
refer the reader to \citet{gratzer:2020} for a full account.

The additional types and terms are then given as follows:
\begin{mathparpagebreakable}
  \inferrule{
    \IsTy[\LockCx{\Gamma}]{A}
  }{
    \IsTy{\Modify{A}}
  }
  \and
  \inferrule{
    \IsTm[\LockCx{\Gamma}]{M}{A}
  }{
    \IsTm{\MkMod{M}}{\Modify{A}}
  }
  \and
  \inferrule{
    \IsTy[\LockCx{\Gamma}[\nu\circ\mu]]{A}
    \\
    \IsTy[\MECx{\Gamma}[\nu]{\Modify[\mu]{A}}]{B}
    \\
    \IsTm[\MECx{\Gamma}[\nu\circ\mu]{A}]{b}{\Sb{B}{\ESb{\Wk}{\MkMod{\mathbf{var}}}}}
    \\
    \IsTm[\LockCx{\Gamma}[\nu]]{a}{\Modify[\mu]{A}}
  }{
    \IsTm{\LetMod<\nu>[\mu]{a}{b}}{\Sb{B}{\ESb{\ISb}{a}}}
  }
  \and
  \inferrule{
    \IsTy[\LockCx{\Gamma}[\nu\circ\mu]]{A}
    \\
    \IsTy[\MECx{\Gamma}[\nu]{\Modify[\mu]{A}}]{B}
    \\
    \IsTm[\MECx{\Gamma}[\nu\circ\mu]{A}]{b}{\Sb{B}{\ESb{\Wk}{\MkMod{\mathbf{var}}}}}
    \\
    \IsTm[\LockCx{\Gamma}[\nu\circ\mu]]{a}{A}
  }{
    \EqTm{
      \LetMod<\nu>[\mu]{\MkMod{a}}{b}
    }{
      \Sb{b}{\ESb{\ISb}{a}}
    }{
      \Sb{B}{\ESb{\ISb}{\MkMod{a}}}
    }
  }
  \and
  \inferrule{
    \IsTy[\LockCx{\Gamma}]{A}
    \\
    \IsTy[\MECx{\Gamma}{A}]{B}
  }{
    \IsTy{\MFn{A}{B}}
  }
  \and
  \inferrule{
    \IsTy[\LockCx{\Gamma}]{A}
    \\
    \IsTm[\MECx{\Gamma}{A}]{b}{B}
  }{
    \IsTm{\lambda{M}}{\MFn{A}{B}}
  }
  \and
  \inferrule{
    \IsTm{f}{\MFn{A}{B}}
    \\
    \IsTm[\LockCx{\Gamma}]{a}{A}
  }{
    \IsTm{f\prn{a}}{\Sb{B}{\ESb{\ISb}{a}}}
  }
  \and
  \inferrule{
    \IsTy[\LockCx{\Gamma}]{A}
    \\
    \IsTm[\MECx{\Gamma}{A}]{b}{B}
    \\
    \IsTm[\LockCx{\Gamma}]{a}{A}
  }{
    \EqTm{(\lambda b)(a)}{\Sb{b}{\ESb{\ISb}{a}}}{\Sb{B}{\ESb{\ISb}{a}}}
  }
  \and
  \inferrule{
    \IsTm{f}{\MFn{A}{B}}
  }{
    \EqTm{f}{\lambda \Sb{f}{\Wk}(\mathbf{var})}{\MFn{A}{B}}
  }
\end{mathparpagebreakable}

\section{Full list of axioms}
\label{sec:axioms}

\univalence*
\intax*
\crispind*
\intdisc*
\intdetectsdisc*
\cubesseparate*
\simplicialstability*
\amazingradj*

\begin{restatable}{axiom}{intop}
  There is an equivalence $\Modify[\OM]{\Int} \to \Int$ which exchanges $0$ for $1$ and $\land$ for
  $\lor$.
\end{restatable}

Define a \emph{finitely-presented $\Int$-algebra} to be a map of bounded distributive lattice $\Int \to X$
where $X$ is equivalent to a bounded distributive lattice of the form $\Int\brk{x_1 \dots
x_n}/\prn{f_1 = g_1 \dots f_m = g_m}$ and $\Int \to X$ is the canonical map. That is, $X$ is freely
generated over $\Int$ by the operations of a bounded distributive lattice, the indeterminates $x_1
\dots x_n$, and subject to the equations $f_i = g_i$. With this notation to hand, we state a duality
axiom due originally to \citet{kock:2014} and proposed in this form by \citet{blechschmidt:2023}.

\begin{restatable}{axiom}{duality}
  If $\Int \to X$ is a finitely presented $\Int$-algebra, the following evaluation map is an
  equivalence of underlying sets:
  \[
    \lambda x, f.\, f\prn{x} : X \Equiv \Int^{\Hom[\Int/]{X}{\Int}} 
  \]
\end{restatable}

\end{document}